\documentclass[apalike]{article}
\usepackage[authoryear]{natbib} 
\usepackage{apalike}
\usepackage[dvips]{graphicx}
\usepackage{amsmath}
\usepackage{amsfonts}
\usepackage{amssymb}
\usepackage{amsthm}
\usepackage{multirow}
\usepackage{array}
\usepackage{tabularx}
\usepackage{mathrsfs}
\usepackage{url}

\newtheorem{lemma}{Lemma}
\newtheorem{theorem}{Theorem}
\newtheorem{assumption}{Assumption}

\newtheorem{proposition}{Proposition}
\newtheorem{definition}{Definition}
\newtheorem{example}{Example}

\makeatletter

\@addtoreset{equation}{section}
\makeatother

\begin{document}
\markboth{Yukihiro Tsuzuki}{Pitman's Theorem, Black--Scholes Equation, and Derivative Pricing for Fundraisers}

\title{Pitman's Theorem, Black--Scholes Equation, and Derivative Pricing for Fundraisers\thanks{We would like to thank Professors Yoshiki Otobe and Bin Xie for their helpful discussions.
This work was supported by JSPS KAKENHI, Grant Number JP19K13737.
}}
\author{Yukihiro Tsuzuki\thanks{Faculty of Economics and Law, Shinshu University, 3-1-1 Asahi, Matsumoto, Nagano 390-8621, Japan; yukihirotsuzuki@shinshu-u.ac.jp}}

\maketitle

\abstract{We propose a financial market model that comprises a savings account and a stock.
The stock price process is modeled as a one-dimensional diffusion,
in which
two types of agents exist: an ordinary investor and a fundraiser who buys or sells stocks as funding activities.
Although the investor information is the natural filtration of the diffusion,
the fundraiser possesses additional information regarding the funding, as well as additional cash flows as a result of the funding.
This concept is modeled using Pitman's theorem for the three-dimensional Bessel process.

Two contributions are presented:
First, the prices of European options for the fundraiser are derived.
Second, a numerical scheme for call option prices in a market with a bubble is proposed,
where multiple solutions exist for the Black--Scholes equation
and the derivative prices are characterized as the smallest nonnegative supersolution.
More precisely, the call option price in such a market is approximated from below by the prices for the fundraiser.
This scheme overcomes the difficulty that stems from the discrepancy that the payoff shows linear growth, whereas the price function shows strictly sublinear growth.
}

{\bf Keywords: Pitman's theorem; Black--Scholes equation; local martingale; financial bubble; non-uniqueness of Cauchy problem; derivative price; funding; Schwarzian derivative} 


\clearpage

\section{Introduction}
We propose a financial market model that comprises a savings account and a stock,
in which two types of agents are considered: an investor and {\it a fundraiser}.

Markets in which different agents have different information have been studied under {\it insider trading},
in which one agent, who is referred to as an insider, possesses extra information over the remaining agents (the reader is referred to, for example, \citet{Higa2004} and \citet{HigaLatorre2011} for surveys).
These markets are modeled using the enlargement technique of filtrations, particularly initial and progressive enlargements.
For a given filtration and random variable,
in initial enlargement, the information of the random variable is added to the original filtration at the initial time.
In progressive enlargement, where the random variable is assumed to be nonnegative, the filtration is enlarged such that the random variable is a stopping time (see, for example, \citet{mansuy2006random} and \citet{protter2005stochastic}).
Although these enlargements have been studied since the 1980s,
a new type of enlargement, namely {\it an enlargement with a stochastic process}, was recently presented by \citet{KchiaProtter2015}.
One example is the enlargement of the natural filtration of the three-dimensional Bessel process with its future infimum; that is, Pitman's theorem (see \citet{Pitman1975} for Pitman's theorem and \citet{SaishoTanemura1990}, \citet{Rauscher1997}, \citet{Takaoka1997}, \citet{mansuy2006random}, and \citet{aksamit2017enlargement} for the generalization thereof).

This study partly aligns with insider trading in the sense that the fundraiser possesses additional information regarding the funding over the investor, who is an ordinary trader.
Furthermore, the fundraiser is supposed to issue new stocks or buy back stocks.
As a result, the information available to the fundraiser is richer than that available to the investor, and the fundraiser has additional cash flows when he or she performs these operations.
This market is also modeled by the enlargement of a filtration through a stochastic process.
The stock price process is modeled as a one-dimensional diffusion under a physical measure
and the investor information is the natural filtration of the diffusion.
The fundraiser views the stock price with the filtration that is enlarged by the future infimum process of another latent diffusion
and is assumed to issue or buy back stocks when the infimum process changes.

More precisely, suppose that the stock price process is determined by $Y=f(X)$ for a differentiable function $f$ and diffusion $X$,
where $f$ is monotonic and $X$ is a transient diffusion with $\lim_{t\rightarrow \infty}X_{t} = \infty$.
Then, the future infimum process $J_{t} = \inf_{u > t} X_{u}$ is finite valued.
We assume that funding occurs when the process $J$ increases strictly.
In the case of $f^{\prime} < 0$, the fundraiser issues new stocks and raises money,
which results in a decrease in stock prices.
In the case of $f^{\prime} > 0$, he or she buys back stocks and spends money,
which results in an increase in stock prices.
If he or she reinvests all money that is raised in the stock, 
the wealth process is calculated as
\begin{eqnarray}
\tilde{Y}_{t} := Y_{t} \left(\frac{f(J_{0})}{f(J_{t})} \right)^{2}.
\end{eqnarray}
The market for the fundraiser is determined as though the stock price is provided by $\tilde{Y}$,
or equivalently, the stock pays dividends $-2df(J_{t})$ only to the fundraiser.
We refer to this market as the {\it hypothetical} market for the fundraiser.

Two contributions are presented, as follows:

First, the prices of European options on the stock for the fundraiser are derived and compared with those for the investor.
The stock price process is modeled using a physical measure and arbitrage opportunities are not excluded.
In such markets, in which no equivalent local martingale measure and only a square-integrable market price of risk exists,
\citet{ruf2012} demonstrated that delta hedging provides the optimal trading strategy in terms of the minimal required initial capital to replicate a given terminal payoff of a European option under the assumption that market models are continuous-time Markovian.
This method generalizes some results in the benchmark approach (\citet{platen2006benchmark}).

We apply the method of \citet{ruf2012} to our problem to obtain the following results:
there exists an equivalent martingale measure in the hypothetical market for the fundraiser; 
a European option price for the fundraiser, which is parameterized by $J_{0}$, namely the level of the stock where the first funding occurs, is higher or lower than that for the investor, depending on the payoff and sign of $f^{\prime}$;
and the fundraiser price converges to the investor price as $J_{0} \rightarrow 0$ under certain conditions.

Second, we propose a numerical scheme to solve the Black--Scholes equation for the prices of call options in a market with a bubble.
In this case, a bubble is defined as a stock price that is a strict local martingale (i.e., a local martingale but not a martingale) under a risk-neutral measure.
Bubbles have been studied extensively in mathematical finance (see, for example, \citet{CoxHobson2005},
\citet{HestonLoewensteinGregory2006},
\citet{JarrowProtterShimbo2007},
\citet{JarrowProtterShimbo2010},
and \citet{Fernholzkaratzas2010}).
Several standard results in derivative pricing theory have failed:
put-call parity does not hold;
the European call option price is neither convex in the stock price nor increasing in maturity;
the price of an American call exceeds that of a European call;
and solutions to Black--Scholes equations are non-unique,
which is investigated in this study. 
In general, derivative prices are solutions to the corresponding partial differential equations or Cauchy problems, which are referred to as Black--Scholes equations in the literature,
the existence and uniqueness of which have been studied in detail.
For example, \citet{HeathSchweizer2000} and \citet{JansonTysk2006} derived sufficient conditions for the stochastic solution to be the unique solution; in this case, the stochastic solution is defined as the conditional expectation of the payoff.
However, multiple solutions may exist in the presence of bubbles,
among which derivative prices are characterized as the smallest nonnegative supersolution (\citet{HestonLoewensteinGregory2006}).
\citet{Tysk2009} found that the uniqueness holds if a payoff function is of strictly sublinear growth,
and \citet{Cetin2018} established a new characterization of the stochastic solution in terms of the unique solution to an alternative Cauchy problem.
Although uniqueness is theoretically recovered, numerical issues remain for a payoff of linear growth.
\citet{EkstromLotstedtSydowTysk2011}
and
\citet{SongYang2015}
provided schemes to approximate the price and demonstrated that the approximate sequences converge.
However, specifying the boundary conditions at a spatial boundary, which is necessary when applying the finite difference method, is difficult.
This is because, as noted in Corollary 6.1 of \citet{Cetin2018}, 
the price is of strictly sublinear growth, whereas the payoff is of linear growth.
This problem is realized as a discontinuity at the corner of the terminal-boundary datum in numerical procedures, as outlined by \citet{SongYang2015}.

We propose an alternative approach that approximates the option price using the prices of the fundraiser to address the above problem.
These prices converge to the price of the investor as $J_{0}$ approaches zero, as mentioned previously,
and satisfy the same Black--Scholes equation with an additional boundary condition.
The newly required boundary values are continuously connected to the terminal condition, thereby overcoming the aforementioned difficulty,
and are calculated as averages of the prices of knock-out forward contracts.
The above approach was further developed in \citet{tsuzuki2024boundary}, in which an asymptotic result was obtained. See also \citet{tsuzuki2024boundary} for numerical experiments on the procedure in this study, as well as the numerical methods mentioned above.

Finally, we provide a technical remark on the expressions of the Girsanov density, which play an important role in this study.
Although this process is generally expressed using a stochastic integral and its quadratic variation,
it can also be represented by a multiplicative functional in the case of one-dimensional diffusion, in which the stochastic integral is replaced with a Markov process and an additive functional.
This form is more general and 
\citet{SaschaGuntherRogers2021} used it to handle cases in which the diffusion is not a solution to a stochastic differential equation, whereas
\citet{Cetin2018} proposed a new class of path transformation, namely a recurrent transformation, 
which extends the $h$-transform of Doob to cases in which $h$ is an absolutely continuous function and $h^{\prime}$ has a singular term with respect to the Lebesgue measure, resulting in the appearance of a singular term in the Girsanov density.
Although the additive functionals in these expressions are related to the generator of the diffusion,
we prefer the use of {\it the Schwarzian derivative},
which was introduced by Schwarz in his study of conformal maps,
to express the additive functional.
In this study, we refer to this process as the Schwarzian process.
This is because the Schwarzian derivative is invariant under a M\"{o}bius transform,
and
the relationship between the investor and fundraiser involves the correspondence of $s$ and $1/s$, where $s$ or $-s$ is the scale function of $X$.
Furthermore, the Schwarzian derivative is considered as a measure of the curvature and deviation from the best approximation by a M\"{o}bius transform (see Chapter 6 of \citet{beals2020explorations}).
The introduction of the Schwarzian process eases the computations, revealing the relationship between the investor and fundraiser more clearly.

In the following section, we briefly review the Schwarzian derivative, define the Schwarzian process, and use these processes to express the diffusion laws.
We describe the financial market for both the investor and fundraiser, in which the growth optimal portfolios (GOPs) are derived, in Section \ref{sec:market}.
In Section \ref{sec:derivative}, 
we derive European option prices for the fundraiser,
compare them with those of the investor,
and demonstrate that the former converges to the latter under certain conditions,
which are summarized in Theorem \ref{them:main} (the main theorem of this study).
In Section \ref{sec:bse}, we propose a numerical scheme for call option prices in markets with bubbles using the convergence result.
Finally, we provide various examples in Section \ref{sec:example}.

\section{Schwarzian derivative and Schwarzian process}
The Schwarzian derivative is ubiquitous
and tends to appear in seemingly unrelated fields of mathematics, including
classical complex analysis, differential equations, one-dimensional dynamics, Teichm\"{u}ller theory, integrable systems, and conformal field theory
(\citet{OvsienkoTabachnikov2009}); however,
it is rarely used in stochastic calculus.
In the following, we briefly review the Schwarzian derivative,
introduce the related stochastic process, which we refer to as the Schwarzian process,
and use it to express the diffusion laws.

\subsection{Schwarzian derivative}
We introduce the pre-Schwarzian and Schwarzian derivatives.
Let $I,I^{\prime}$ be open intervals in $\mathbb{R}$,
and
$\mathcal{D}(I,I^{\prime})$ be the set of all functions $f: I \longrightarrow I^{\prime}$
such that $f \in C^{3}(I,I^{\prime})$ and $f^{\prime}(x)\neq 0$ for all $x \in I$.
In the following, we omit the domain $\mathcal{D}(I,I^{\prime})$ if it is clear and we denote $\mathcal{D}(I)$ if $I = I^{\prime}$.
\begin{definition}
The pre-Schwarzian derivative $T_{f}$
and Schwarzian derivative $S_{f}$ of $f$
are respectively defined as follows:
\begin{eqnarray}
T_{f} &:=& \frac{f^{\prime\prime}}{f^{\prime}} = (\log f^{\prime})^{\prime}, \\
S_{f} &:=& T_{f}^{\prime} - \frac{1}{2} T_{f}^{2}
  = \left(\frac{f^{\prime\prime}}{f^{\prime}} \right)^{\prime}
  - \frac{1}{2}\left(\frac{f^{\prime\prime}}{f^{\prime}}\right)^{2}
  = \frac{f^{\prime\prime\prime}}{f^{\prime}}
  - \frac{3}{2}\left(\frac{f^{\prime\prime}}{f^{\prime}}\right)^{2}.
\end{eqnarray}
\end{definition}

We list several basic properties (see, for example, Chapter 6 of \citet{beals2020explorations}):
\begin{itemize}
\item[(T1)] $T_{f} = 0$ if and only if $f$ is an affine transform.
\item[(T2)] $T_{g\circ f} = (T_{g}\circ f)\cdot f^{\prime}+T_{f}$.
\item[(T3)] In particular, $T_{g\circ f} = T_{f}, T_{f\circ g} = (T_{f}\circ g)\cdot g^{\prime}$ if $g$ is an affine transform.
\item[(T4)] If $T_{f}=T_{g}$, $f$ is an affine transform of $g$.
\item[(T5)] $T_{1/f} = T_{f} - 2 f^{\prime}/f$.
\end{itemize}  
Furthermore:
\begin{itemize}
\item[(S1)] $S_{f} = 0$ if and only if $f$ is a M\"{o}bius transform.
\item[(S2)] $S_{g\circ f} = (S_{g}\circ f)\cdot (f^{\prime})^{2}+S_{f}$.
\item[(S3)] In particular, $S_{g\circ f} = S_{f}, S_{f\circ g} = (S_{f}\circ g)\cdot (g^{\prime})^{2}$ if $g$ is a M\"{o}bius transform.
\item[(S4)] If $S_{f}=S_{g}$, $f$ is a M\"{o}bius transform of $g$.
\item[(S5)] $h:=1/\sqrt{f^{\prime}}$ satisfies
$h^{\prime} = -\frac{1}{2} T_{f} h$
and $h^{\prime\prime} = -\frac{1}{2} S_{f} h$.
\end{itemize}

\subsection{Schwarzian process}
Let $X$ be a continuous semimartingale that has values in $I$.
For $f \in \mathcal{D}(I,I^{\prime})$, we define a stochastic process that we refer to as the Schwarzian process of $f$, as follows:
\begin{definition}
The Schwarzian process $\mathcal{S}^{f}(X)$ of $f \in \mathcal{D}(I,I^{\prime})$ 
is defined as
\begin{eqnarray}
  \mathcal{S}_{t}^{f}(X)
  = \sqrt{\frac{f^{\prime}(X_{0})}{f^{\prime}(X_{t})}}
  \exp \left(\frac{1}{4} \int_{0}^{t} S_{f}(X_{u}) d\left<X,X\right>_{u} \right). \label{eq:schwarizan_process}
\end{eqnarray}
\end{definition}

The following properties correspond to (S1) to (S5):
\begin{itemize}
\item[(1)] If $\varphi$ is a M\"{o}bius transform; 
more precisely, $\varphi(x) = \frac{ax+b}{cx+d}$ for $ad-bc \neq 0$,
we obtain
\begin{eqnarray}
\mathcal{S}^{\varphi}(X) = \frac{cX+d}{cX_{0}+d}.
\end{eqnarray}
\item[(2)] The composition of two functions corresponds to the multiplication of two Schwarzian processes:
\begin{eqnarray}
\mathcal{S}_{t}^{f\circ g}(X)
=
\mathcal{S}_{t}^{f}(g(X))\mathcal{S}_{t}^{g}(X). \label{eq:composition}
\end{eqnarray}
\item[(3)] In particular, if $f=\varphi$ or $g=\varphi$ for a M\"{o}bius transform $\varphi$, as above,
\begin{eqnarray}
\mathcal{S}_{t}^{\frac{ag+b}{cg+d}}(X)
=
\frac{cg(X_{t})+d}{cg(X_{0})+d} \mathcal{S}_{t}^{g}(X)
\end{eqnarray}
and
\begin{eqnarray}
\mathcal{S}_{t}^{f \circ \varphi}(X)
=
\mathcal{S}^{f}\left(\frac{aX+b}{cX+d}\right)\frac{cX+d}{cX_{0}+d}.
\end{eqnarray}
\item[(4)] If $\mathcal{S}^{f}(X)=\mathcal{S}^{g}(X)$,
$T_{f}(X_{t})=T_{g}(X_{t})$ for $d\left<X,X\right>_{t}$-a.e. almost surely,
which follows from property (5) below.
\item[(5)] The application of Ito's formula to $h=1/\sqrt{f^{\prime}}$ 
yields
\begin{eqnarray}
  h(X_{t}) = h(X_{0})
  -\frac{1}{2} \int_{0}^{t} h(X_{u}) \left(T_{f}(X_{u})dX_{u}+\frac{1}{2}S_{f}(X_{u})d\left<X,X\right>_{u} \right),
\end{eqnarray}
and
a Schwarzian process is an exponential semimartingale:
\begin{eqnarray}
  \mathcal{S}_{t}^{f}(X)
  &=& \exp \left(-\frac{1}{2}\int_{0}^{t} T_{f}(X_{u})dX_{u}-\frac{1}{8} \int_{0}^{t} T_{f}(X_{u})^{2} d\left<X,X\right>_{u} \right)  \nonumber\\
&=& 1 - \int_{0}^{t} \frac{1}{2}T_{f}(X_{u}) \mathcal{S}_{u}^{f}(X) dX_{u}.
\end{eqnarray}
In particular, if $X$ is a local martingale,
so is the Schwarzian process $\mathcal{S}^{f}(X)$.
\end{itemize}

\begin{example}
If $f(x)=(x-\xi)^{\alpha}$ for $\alpha \neq 0$ and $f(x)=\log (x-\xi)$ for $\alpha=0$,  
we obtain 
\begin{eqnarray}
T_{f}(x) &=& \frac{\alpha-1}{x-\xi},\\
S_{f}(x) &=& \frac{1-\alpha^{2}}{2(x-\xi)^{2}},\\
\mathcal{S}_{t}^{f}(X)
&=& \left(\frac{X_{0}-\xi}{X_{t}-\xi}\right)^{\alpha/2-1/2}
\exp \left(\frac{1}{8}\left(1-\alpha^{2}\right)
\int_{0}^{t} \frac{d \left<X,X\right>_{u}}{(X_{u}-\xi)^{2}}\right).
\end{eqnarray}
We denote this process as $\mathcal{S}^{(\alpha,\xi)}$
and $\mathcal{S}^{(\alpha)}=\mathcal{S}^{(\alpha,0)}$.
Note that, if $X$ is a Wiener process,
the process $\mathcal{S}^{(-2\nu)}(X)$ is a nonnegative local martingale, which is the Radon--Nikodym density with respect to the Wiener measure of the probability measure, where $X$ is a Bessel process with index $\nu$.
\end{example}

\begin{example}
Let $\xi_{1},\cdots,\xi_{n}$ and $\alpha_{1},\cdots,\alpha_{n}$ be real numbers
such that $0<\pi \alpha_{j}<2\pi$, and $\sum_{j=1}^{n}(\alpha_{j}-1)=-2$.
Then, using the Schwarz--Christoffel formula,
\begin{eqnarray}
f(x) = \int_{1}^{x} \Pi_{j=1}^{n}(\xi-\xi_{j})^{\alpha_{j}-1} d\xi
\end{eqnarray}
is the conformal mapping from the upper half-plane to a polygon,
and 
\begin{eqnarray}
  T_{f}(x) &=& \sum_{j=1}^{n} \frac{\alpha_{j}-1}{x-\xi_{j}}, \\
  S_{f}(x) &=& \sum_{j=1}^{n} \left(\frac{1-\alpha_{j}^{2}}{2(x-\xi_{j})^{2}} + \frac{\beta_{j}}{x-\xi_{j}}\right),\\
  \mathcal{S}_{t}^{f}(X) &=& \prod_{j=1}^{n} \mathcal{S}_{t}^{(\alpha_{j},\xi_{j})}(X)\exp \left(\frac{1}{4}\int_{0}^{t}\frac{\beta_{j}}{X_{u}-\xi_{j}} d \left<X,X\right>_{u} \right),
\end{eqnarray}
where $\beta_{j}$ are real numbers (see Theorem 6.4.2 of \citet{beals2020explorations}).
\end{example}

\subsection{Laws of diffusions}
\label{sec:space}
Throughout this study, we work on the sample space $\Omega=C^{\pm}([0,\infty)) \times \mathbb{R}$,
where $C^{\pm}([0,\infty))$ denotes the space of continuous paths $\omega:[0,\infty) \rightarrow \mathbb{R} \cup \{\pm \infty\}$
and $\pm \infty$ are assumed to be absorbing points that are
equipped with the coordinate process $(X^{*},J_{0}^{*})$; that is,
$(X_{t}^{*},J_{0}^{*})(\omega,j)=(\omega(t),j)$ for $(\omega,j)\in \Omega$.
We introduce the following further processes in this space:
$J_{t}^{*} = J_{0}^{*} \vee \sup_{u \le t} X_{u}^{*}$ for $t>0$,
$X_{t} = 2J_{t}^{*} - X_{t}^{*}$,
and $J_{t}=\inf_{u\ge t}X_{u}$;
a probability measure $W_{x}$, under which $X$ is a Wiener process starting from $x \in \mathbb{R}$;
sigma algebras
$\mathcal{F}_{t}^{0} = \bigcap_{\varepsilon > 0} \sigma(X_{u}; u \le t+\varepsilon)$
and $\tilde{\mathcal{F}}_{t}^{0} = \bigcap_{\varepsilon > 0} \sigma(X_{u}^{*},J_{0}^{*}; u \le t+\varepsilon)$;
and a random time $\Lambda_{t}^{*}$:
\begin{eqnarray}
\Lambda_{t}^{*} := \sup \{u \in (t,\infty] : X_{u}^{*} = J_{u}^{*} = J_{t}^{*} \},
\end{eqnarray}
where $\Lambda_{t}^{*}=\infty$ if and only if $\lim_{u\rightarrow \infty}X_{u}^{*} = J_{t}^{*}$, and $\Lambda_{t}^{*}=\mathfrak{T}$, which is a time beyond the horizon, if and only if $\{\}=\emptyset$.

Note that
$\{J_{t} = J_{t}^{*}\}$
does not belong to $\tilde{\mathcal{F}}_{t}^{0}$, but belongs to $\tilde{\mathcal{F}}_{\infty}^{0}$.
Pitman's theorem demonstrates that $J_{t} = J_{t}^{*}$ almost surely holds for any $t \ge 0$
if $X$ is a three-dimensional Bessel process,
and $J_{t}$ and $J_{t}^{*}$ need not be distinguished when working with such a probability measure.
However, not all probability measures have this property. 
The remark following Proposition \ref{prop:skorokhod_sde}
describes the existence of a probability measure under which $J_{\infty}^{*} < J_{\infty}=\infty$ almost surely holds.

We consider the probabilities in the sample space $\Omega$
such that $X$ is a diffusion that takes values in the interval $I$,
for which we use open bounded intervals $I_{n}$ that increase to $I$ with $\bar{I}_{n} \subseteq I_{n+1}$.
The stopping time $T_{n}$ with respect to $(\mathcal{F}_{t}^{0})$ is defined as
\begin{eqnarray}
T_{n} = \inf \{t \ge 0 : X_{t} \notin I_{n} \}
\end{eqnarray}
for $n=1,2,\cdots$,
and $T_{\infty} := \lim_{n \rightarrow \infty} T_{n}$.

Finally, as a general notation, for a stopping time $T$ and process $\gamma$, the stopped process is denoted by $\gamma^{T}$,
and for a probability measure $P$, the expectation operator is also denoted by $P$.
For $f \in \mathcal{D}$, the operator $\mathcal{G}^{f}$ denotes
the second-order differential operator:
\begin{eqnarray}
\mathcal{G}^{f} := \frac{1}{2} \frac{d^{2}}{dx^{2}}-\frac{1}{2}T_{f}(x) \frac{d}{dx}.
\end{eqnarray}

\subsubsection{Laws of general diffusions}
For $x \in I, s \in \mathcal{D}(I,I^{\prime})$, as the Schwarzian process $\mathcal{S}^{s}(X)$ is a nonnegative local martingale on $(\Omega,(\mathcal{F}_{t}^{0}),W_{x})$,
a probability measure $P_{x}^{s}$
on $\bigvee_{n} \mathcal{F}_{T_{n}}^{0}$ can be introduced such that
\begin{eqnarray}
  P_{x}^{s}|_{\mathcal{F}_{T_{n}}^{0}} = \mathcal{S}_{T_{n}}^{s}(X) W_{x}|_{\mathcal{F}_{T_{n}}^{0}}
\end{eqnarray}
for each $n$.
Subsequently,
\begin{eqnarray}
B_{t} = X_{t} + \int_{0}^{t} \frac{1}{2}T_{s}(X_{u}) du,\; t < T_{\infty} \label{eq:s_sde}
\end{eqnarray}
is a $P_{x}^{s}$-Brownian motion according to the Girsanov theorem,
$X$ is a diffusion with the generator $\mathcal{G}^{s}$,
and $s(X)$ is on its natural scale under $P_{x}^{s}$.
If the probability measure $P_{x}^{f}$ is similarly introduced, 
for each $n$, we obtain
\begin{eqnarray}
  P_{x}^{f}|_{\mathcal{F}_{T_{n}}^{0}} = \frac{\mathcal{S}_{T_{n}}^{f}(X)}{\mathcal{S}_{T_{n}}^{s}(X)} P_{x}^{s}|_{\mathcal{F}_{T_{n}}^{0}}
  = \mathcal{S}_{T_{n}}^{f\circ s^{-1}}(s(X)) P_{x}^{s}|_{\mathcal{F}_{T_{n}}^{0}},\label{eq:density}
\end{eqnarray}
where (\ref{eq:composition}) is used.

\subsubsection{Laws of transient diffusions and their future infimum}
Hereafter, we suppose $x > 0$ and the following assumption on $s$,
which ensures that
$\inf_{t \ge 0} X_{t} > 0$ and $\sup_{t \ge 0} X_{t} = \infty$, $P_{x}^{s}$-a.s.
according to Feller's test,
and $X$ does not explode in finite time $P_{x}^{s}$-a.s.
(see, for example, Theorem 5.29 in Chapter 5 of \citet{karatzas}),
thereby resulting in $P_{x}^{s}[T_{\infty} = \infty]=1$:
\begin{assumption}
\label{ass:s}
The function $s : (0,\infty) \longrightarrow (0,\infty)$ satisfies
\begin{eqnarray}
s(0) = +\infty,\; s(+\infty) = 0,\; s^{\prime} < 0 \label{ass:s1}
\end{eqnarray}
and
\begin{eqnarray}
  \lim_{x \rightarrow \infty}\int_{1}^{x} \frac{s(x)-s(\xi)}{s^{\prime}(\xi)} d\xi = \infty.\label{ass:s2}
\end{eqnarray}
\end{assumption}

We denote $(\mathcal{F}_{t}^{s})$ as the $P_{x}^{s}$-augmented natural filtration of $B$.
Thus, it is also the $P_{x}^{s}$-augmented natural filtration of $X$
because $B_{t}$ is a functional of $X_{u}, u \le t$,
and the pathwise uniqueness holds for the stochastic differential equation in the form of (\ref{eq:s_sde}).
Under $P_{x}^{s}$, the process $X$ is a transient diffusion with $P_{x}^{s}[\lim_{t \rightarrow \infty} X_{t} = \infty]=1$,
and $J$ is an $\mathbb{R}_{+}$-valued process.
We define $(\mathcal{F}_{t}^{s,J})$ as the $P_{x}^{s}$-augmentation of $(\mathcal{F}_{t}^{0} \vee \sigma(J_{t}))$
and extend $P_{x}^{s}$ such that $P_{x}^{s}[J_{0}^{*}=J_{0}]=1$.
Then, 
\begin{eqnarray}
\tilde{B}_{t} = B_{t}-2J_{t} - \int_{0}^{t} (s^{\prime}/s)(X_{u}) du
\end{eqnarray}
is a $(\mathcal{F}_{t}^{s,J})$-Brownian motion (see Section 5.7 of \citet{aksamit2017enlargement}),
and the semimartingale decompositions of $X$ for each filtration are expressed as
\begin{eqnarray}
  X_{t}
  &=& B_{t} - \int_{0}^{t} \frac{1}{2}T_{s}(X_{u}) du \\
  &=& \tilde{B}_{t} - \int_{0}^{t} \frac{1}{2}T_{1/s}(X_{u}) du + 2J_{t},\label{eq:sde_xj_s}
\end{eqnarray}
where property (T5) is used.

The case of $s(x)=1/x$, which satisfies (\ref{ass:s1}) and (\ref{ass:s2}),
corresponds to the law of the three-dimensional Bessel process.
If $P_{x}^{(1/2)}$ is the law that is defined by
\begin{eqnarray}
P_{x}^{(1/2)}|_{\mathcal{F}_{T_{n}}^{0}} = \frac{X_{T_{n}}}{x} W_{x}|_{\mathcal{F}_{T_{n}}^{0}},
\end{eqnarray}
$(X_{t})_{t \in [0,\infty)}$ is a three-dimensional Bessel process under $P_{x}^{(1/2)}$.
According to Pitman's theorem, as presented in Theorem 3.5 in Chapter VI of \citet{revuzyor},
the dual process $X_{t}^{*}$ is a Brownian motion.

We introduce the counterpart of $\mathcal{S}$ with respect to the enlarged filtration $(\tilde{\mathcal{F}}_{t}^{0})$:
\begin{eqnarray}
\tilde{\mathcal{S}}_{t}^{f}(X)
:= \frac{f^{\prime}(J_{t}^{*})}{f^{\prime}(J_{0}^{*})}  \mathcal{S}_{t}^{f}(2J^{*}-X^{*}),
\end{eqnarray}
which is a $P_{x}^{(1/2)}$-a.s. real-valued process that is adapted to $(\tilde{\mathcal{F}}_{t}^{0})$.
Subsequently, $\tilde{\mathcal{S}}^{f}(X)$ is a local martingale on $(\Omega,(\tilde{\mathcal{F}}_{t}^{0}),P_{x}^{(1/2)})$:
\begin{eqnarray}
\tilde{\mathcal{S}}_{t}^{f}(X)
= 1 + \int_{0}^{t} \frac{1}{2} T_{f}(2J_{u}^{*}-X_{u}^{*}) \tilde{\mathcal{S}}_{u}^{f}(X) dX^{*}_{u}.
\end{eqnarray}
We introduce a probability measure $\tilde{P}_{x}^{f}$ on $\bigvee_{n} \tilde{\mathcal{F}}_{T_{n}}^{0}$ such that
\begin{eqnarray}
\tilde{P}_{x}^{f}|_{\tilde{\mathcal{F}}_{T_{n}}^{0}} = \tilde{\mathcal{S}}_{T_{n}}^{f}(X) P_{x}^{(1/2)}|_{\tilde{\mathcal{F}}_{T_{n}}^{0}}.
\end{eqnarray}
According to Girsanov's theorem,
\begin{eqnarray}
\tilde{\beta}_{t} := -\left(X^{*}_{t} - \int_{0}^{t}\frac{1}{2}T_{f}(2J_{u}^{*}-X^{*}_{u}) du \right) \label{eq:bm_fxj}
\end{eqnarray}
is a $((\tilde{\mathcal{F}}_{t}^{0}),\tilde{P}_{x}^{f})$-Brownian motion up to $T_{\infty}$,
and hence, $X$ satisfies
\begin{eqnarray}
X_{t} = \tilde{\beta}_{t} - \int_{0}^{t} \frac{1}{2}T_{f}(X_{u}) du + 2J_{t}^{*}\label{eq:sde_xj_f}
\end{eqnarray}
up to $T_{\infty}$.

Conversely, for a given law of $J_{0}^{*}$ on $(0,x)$
and Brownian motion $\tilde{\beta}$ with $\tilde{\beta}_{0}=2J_{0}^{*}-x$,
$(X^{*},J^{*})$ can be constructed through a stochastic differential equation of the Skorokhod type using the following proposition with $(\chi,l)=(J^{*}-X^{*},J^{*})$:
\begin{proposition}
\label{prop:skorokhod_sde}
For $\chi_{0} \ge 0,l_{0} > 0$ and a Brownian motion $\gamma$ with $\gamma_{0} = \chi_{0}-l_{0}$,
there exists a unique pair $(\chi,l)$ of continuous functions that are adapted to $\gamma$
such that
$\chi$ is nonnegative,
$l$ is nondecreasing, 
\begin{eqnarray}
l_{t} = l_{0}+\int_{0}^{t} 1_{\{\chi_{u}=0\}} dl_{u},
\end{eqnarray}
and $(\chi,l)$ satisfies the stochastic differential equation of the Skorokhod type:
\begin{eqnarray}
\chi_{t} = \gamma_{t} - \int_{0}^{t} \frac{1}{2}T_{f}(\chi_{u}+l_{u}) du + l_{t} \label{eq:skorokhod}
\end{eqnarray}
for $t < \zeta := \inf \{t : \chi_{t}+l_{t} = \infty\}$.
Furthermore, suppose that there exists $C>0$ such that 
\begin{eqnarray}
|T_{f}(x)| < C x 
\end{eqnarray}
for $x>l_{0}$.
Subsequently, $\zeta=\infty$ almost surely.
\end{proposition}
Proposition \ref{prop:skorokhod_sde} can be proven in the same manner 
as in \citet{SaishoTanemura1990}, who considered a more general case with $\chi_{0} = l_{0} = 0$ under the assumption that the coefficients satisfy global Lipschitz conditions.
We also refer to an equation such as (\ref{eq:sde_xj_f})
as a stochastic differential equation of the Skorokhod type.
Whereas $P_{x}^{(1/2)}[J_{t}=J_{t}^{*}]=P_{x}^{(1/2)}[\Lambda_{t}^{*}<\infty]=1$ for all $t \ge0$,
we may obtain $\tilde{P}_{x}^{f}[\Lambda_{t}^{*}=\mathfrak{T}] > 0$,
which is the case, for example, if $T_{f}$ is a negative constant; that is, $\chi-l$ is a Brownian motion with positive drift.

Next, we discuss the equivalence between $P_{x}^{(1/2)}$ and $\tilde{P}_{x}^{f}$. The equivalence between two probability measures was investigated extensively
by \cite{MijatovicUrusov2012_FS}, \citet{MijatovicUrusov2012_PR}, 
and especially \citet{CarrFisherRuf2014}, with great generality.
As opposed to applying these results to our case,
we provide sufficient conditions for the equivalence for completeness.
\begin{assumption}
\label{ass:equivalence}
There exists $C > 0$ such that $|T_{f}(x)| < C x$ for any $x>1$.
\end{assumption}
\begin{lemma}
\label{lem:equivalence}
Suppose that $f$ satisfies Assumption \ref{ass:equivalence}.
Thus, $\tilde{P}_{x}^{f}[T_{\infty}=\infty]=1$ holds,
in which case the two probability measures $P_{x}^{(1/2)}$ and $\tilde{P}_{x}^{f}$ are equivalent on $\tilde{\mathcal{F}}_{T}^{0}$ for any $T>0$.
\end{lemma}
\begin{proof}
By applying Proposition \ref{prop:skorokhod_sde} with $(\chi,l)=(X-J^{*},J^{*})$ and $\gamma = \tilde{\beta}$,
we determine $0 < J_{0}^{*} \le J_{t}^{*} \le X_{t}$,
$\tilde{P}_{x}^{f}$-a.s.
We obtain $\tilde{P}_{x}^{f}[T_{\infty}=\infty]=1$.
Let $\tilde{M}$ be the reciprocal of $\tilde{S}^{f}(X)$.
Then, $\tilde{M}$ is a local martingale with respect to $((\tilde{\mathcal{F}}_{t}^{0}),\tilde{P}_{x}^{f})$
with a localizing sequence $T_{n}$.
Furthermore, $\tilde{M}$ is a martingale because
\begin{eqnarray}
\tilde{P}_{x}^{f}[\tilde{M}_{T}]
\ge
\lim_{n\rightarrow \infty} \tilde{P}_{x}^{f}[\tilde{M}_{T \wedge T_{n}} : T < T_{n}]
=
\lim_{n\rightarrow \infty} P_{x}^{(1/2)}[T < T_{n}]
=
1.
\end{eqnarray}
\end{proof}

\section{Financial market}
\label{sec:market}
We consider a financial market that comprises two primary security accounts, namely a savings account and a stock, in the probability space $(\Omega,(\mathcal{F}_{t}^{s}),P_{x}^{s})$, where $s$ satisfies Assumption \ref{ass:s}.
For simplicity, the interest rate is assumed to be zero,
and the stock price is modeled as $Y=f(X)$.
We provide the following assumption regarding $f$
to ensure that $X$ does not explode to infinity in finite time $P_{x}^{f}$-a.s.:
\begin{assumption}
\label{ass:f}
The function $f \in \mathcal{D}((0,\infty))$ satisfies the following: 
\begin{itemize}
\item[(1)] If $f^{\prime} < 0$, then $f(+\infty) = 0$ and
\begin{eqnarray}
\lim_{x \rightarrow \infty}\int_{1}^{x} \frac{f(x)-f(\xi)}{f^{\prime}(\xi)} d\xi = \infty.
\end{eqnarray}
\item[(2)] If $f^{\prime} > 0$, $f(0+) = 0$ and $f(+\infty) = +\infty$.
\end{itemize}
\end{assumption}

We assume that two types of agents exist: an investor and a fundraiser.
The investor is an ordinary trader, whereas
the fundraiser is supposed to issue new stocks or buy back stocks.
As a result, the information available to the fundraiser is richer than that available to the investor and he or she has additional cash flows when these operations are performed.
Section \ref{sec:investor_market} describes the market for the investor, with a particular focus on the GOP,
whereas Section \ref{sec:fundraiser_market} explains the market for the fundraiser.

\subsection{Financial market for investor}
\label{sec:investor_market}
In the filtered probability space $(\Omega,(\mathcal{F}_{t}^{s}),P_{x}^{s})$,
the dynamics of $Y$ is expressed as
\begin{eqnarray}
Y_{t} = Y_{0} + \int_{0}^{t} Y_{u} \frac{f^{\prime}}{f}(X_{u}) \left( dB_{u} + \theta(X_{u}) du \right),
\end{eqnarray}
where $\theta$ is the market price of risk: 
\begin{eqnarray}
\theta(x) = \frac{1}{2}(T_{f}-T_{s})(x).
\end{eqnarray}
In this market, 
the investor makes decisions to buy or sell stocks based on the natural information of $Y$; that is, $(\mathcal{F}_{t}^{s})$.
A pair of $(\mathcal{F}_{t}^{s})$-predictable processes $(\delta^{0},\delta^{1})$
expresses the trading strategy:
$\delta^{0}$ is the number of units of the savings account and $\delta^{1}$ is that of the stock.
We assume that all trading strategies are self-financing;
that is, the value $V$ of the corresponding portfolio is determined by
\begin{eqnarray}
  V_{t} = \delta_{t}^{0} + \delta_{t}^{1} Y_{t}
  = V_{0} + \int_{0}^{t} \delta_{u}^{1} dY_{u},
\end{eqnarray}
where the stochastic integral is assumed to be well defined.
We further assume that $V_{0} > 0$ for all trading strategies 
and consider the portfolio up to when $V$ reaches $0$.
Subsequently, the trading strategy is identified by the fraction $\delta^{1}S / V$ of the stock.
However, it is convenient to identify a portfolio using $\xi$,
which is a predictable process with $\int_{0}^{t} \xi_{u}^{2} du < \infty$,
and is related to $(\delta^{0},\delta^{1})$, using the equation
\begin{eqnarray}
\xi_{t} = \frac{\delta_{t}^{1}Y_{t}}{\delta_{t}^{0} + \delta_{t}^{1} Y_{t}}\frac{f^{\prime}(X_{t})}{f(X_{t})}.
\end{eqnarray}
This is the fraction of the stock multiplied by its volatility.
We denote the class of this type of process $\xi$ as $\mathcal{L}^{2}$
and the corresponding portfolio value as $V^{\xi}$.
The log price of $V^{\xi}$ is
\begin{eqnarray}
\log V_{t}^{\xi}
= \log V_{0}^{\xi}
+ \int_{0}^{t} \xi_{u} dB_{u} + \int_{0}^{t} \frac{1}{2}(-\xi_{u}^{2}+2\theta(X_{u}) \xi_{u}) du.
\end{eqnarray}

Next, we introduce the GOP,
which is characterized as the portfolio that maximizes the drift term of the log price.
\begin{definition}
The GOP is a portfolio with a value $G=V^{\eta}, \eta \in \mathcal{L}^{2}$ that satisfies any $\xi \in \mathcal{L}^{2}$
\begin{eqnarray}
-\eta_{t}^{2}+2\theta(X_{t}) \eta_{t}
\ge
-\xi_{t}^{2}+2\theta(X_{t}) \xi_{t}
\end{eqnarray}
for all $t$ almost surely.
\label{def:gop}
\end{definition}
We always consider a GOP with $G_{0}=1$.
The following proposition can be demonstrated using straightforward computations:
\begin{proposition}
The GOP for the investor is
\begin{eqnarray}
G_{t} = V_{t}^{\theta} = \frac{\mathcal{S}_{t}^{s}(X)}{\mathcal{S}_{t}^{f}(X)}
= \sqrt{\frac{s^{\prime}(X_{0})}{s^{\prime}(X_{t})}}
\sqrt{\frac{f^{\prime}(X_{t})}{f^{\prime}(X_{0})}}
\exp\left(\frac{1}{4} \int_{0}^{t} (S_{s}-S_{f})(X_{u}) du \right), \label{eq:gop}
\end{eqnarray}
and we obtain $P_{x}^{s}$-a.s. that $0 < G_{t} < \infty$ for all $t \in [0,\infty)$.
\end{proposition}

Note that a GOP is also characterized by the following equivalent conditions:
\begin{itemize}
\item[(i)] $V^{\xi}/G$ is a local martingale for any $\xi \in \mathcal{L}^{2}$ until $V^{\xi}$ reaches $0$.
\item[(ii)] $Z:=1/G$ and $Y/G$ are local martingales.
\end{itemize}

\subsection{Financial market for fundraiser}
\label{sec:fundraiser_market}
We consider the stock price $Y$ with the enlarged filtration $(\mathcal{F}_{t}^{s,J})$ in the probability space $(\Omega,\mathcal{F}_{\infty}^{s,J},P_{x}^{s})$.
The semimartingale decomposition of $Y$ with respect to $(\mathcal{F}_{t}^{s,J})$ is
\begin{eqnarray}
Y_{t}
= Y_{0} + \int_{0}^{t} Y_{u} \frac{f^{\prime}}{f}(X_{u}) \left( d\tilde{B}_{u} + \tilde{\theta}(X_{u}) du \right)+2(f(J_{t})-f(J_{0})),\label{eq:sde_yj}
\end{eqnarray}
where 
\begin{eqnarray}
\tilde{\theta}(x) = \frac{1}{2}(T_{f}-T_{1/s})(x).
\end{eqnarray}

Let $L_{t}=2(f(J_{0})-f(J_{t}))$ and $n$ be a positive adapted process, which we interpret as the number of stocks.
We assume that
\begin{eqnarray}
\int_{0}^{t} Y_{u} dn_{u} = \int_{0}^{t} n_{u} dL_{u}.\label{eq:ass_L}
\end{eqnarray}
The amount of money that is raised by the fundraiser for each stock when funding occurs is $dL_{t}=Y_{t} dn_{t}/n_{t}$.
More precisely, if $f^{\prime} < 0$, he or she issues $dn_{t}$ units of new stocks and raises $Y_{t}dn_{t}=n_{t}dL_{t}$ money.
If $f^{\prime} > 0$, he or she buys back $-dn_{t}$ units of stocks
and spends $-Y_{t}dn_{t}=-n_{t}dL_{t}$ money.
This can be interpreted as causing a rise or fall of $|dL_{t}|$ in the stock price, respectively.
Thus, the singular term in (\ref{eq:sde_yj}), namely $L$,
is considered as the result of artificial actions.
In contrast, the aggregate market price 
\begin{eqnarray}
n_{t}Y_{t} &=& n_{0}Y_{0} + \int_{0}^{t} n_{u} dY_{u} + \int_{0}^{t} Y_{u} dn_{u} \nonumber\\
&=& n_{0}Y_{0} + \int_{0}^{t} n_{u} d(Y_{u}+L_{u})
\end{eqnarray}
has no singular term if and only if Assumption (\ref{eq:ass_L}) holds,
which means that the aggregate market price is not manipulated.

The fundraiser can use the raised money $dL_{t}$ to construct a portfolio.
We consider the strategy that is characterized by $(\delta^{0},\delta^{1})$:
\begin{eqnarray}
  \delta_{t}^{0} &=& \delta_{0}^{0} + \int_{0}^{t} (1-\kappa_{u}) \delta_{u}^{1} dL_{u}, \\
  \delta_{t}^{1} &=& \delta_{1}^{1} + \int_{0}^{t} \kappa_{u} \delta_{u}^{1} \frac{dL_{u}}{Y_{u}},
\end{eqnarray}
where $\kappa$ is a $[0,1]$-valued predictable process,
which is the ratio for allocating the money $dL_{t}$ to its own stock at time $t$.
The value $W$ of the corresponding portfolio is 
\begin{eqnarray}
  W_{t} = W_{0} \exp \left(\int_{0}^{t} \frac{\pi_{u}}{Y_{u}} (dY_{u}+dL_{u}) - \frac{1}{2} \int_{0}^{t} \left(\frac{\pi_{u}}{Y_{u}}\right)^{2} d\left<Y,Y\right>_{u} \right),  \label{eq:reinvestment}
\end{eqnarray}
where $\pi_{t} := \delta_{t}^{1} Y_{t}/W_{t}$.
In particular, the substitution of $\kappa=1$, $\delta_{0}^{0}=0$ and $W_{0}=Y_{0}$ leads to
\begin{eqnarray}
W_{t} = Y_{t} \left(\frac{f(J_{0})}{f(J_{t})} \right)^{2}.\label{eq:fullreinvestment}
\end{eqnarray}
We assume that the fundraiser uses this strategy.
Then, the market for the fundraiser, which we refer to as the hypothetical market, comprises 
a savings account and risky security account $\tilde{Y}$:
\begin{eqnarray}
\tilde{Y}_{t}
:= Y_{t} \left(\frac{f(J_{0})}{f(J_{t})} \right)^{2}
= Y_{0} + \int_{0}^{t} \tilde{Y}_{u} \frac{f^{\prime}}{f}(X_{u}) \left( d\tilde{B}_{u} + \tilde{\theta}(X_{u}) du \right).
\end{eqnarray}
The market price of risk is $\tilde{\theta}(X_{t})$.

We consider the GOP for the fundraiser.
The argument is almost the same as that in Section \ref{sec:investor_market},
and we use the notation with the accent $\tilde{} \;$, such as $\tilde{V}^{\xi}$.
We obtain the following proposition:
\begin{proposition}
The GOP for the fundraiser is
\begin{eqnarray}
\tilde{G}_{t}
= \tilde{V}_{t}^{\tilde{\theta}}
= \sqrt{\frac{(1/s)^{\prime}(X_{0})}{(1/s)^{\prime}(X_{t})}\frac{f^{\prime}(X_{t})}{f^{\prime}(X_{0})}}
\frac{(1/s)^{\prime}(J_{t})}{(1/s)^{\prime}(J_{0})}
\frac{f^{\prime}(J_{0})}{f^{\prime}(J_{t})}
\mathrm{e}^{\frac{1}{4} \int_{0}^{t} (S_{s}-S_{f})(X_{u}) du}.
\end{eqnarray}
\end{proposition}

\section{Derivative pricing}
\label{sec:derivative}
In this section, we discuss pricing derivatives in the market model that was considered in Section \ref{sec:market},
where $P_{x}^{s}$ is regarded as a physical measure.
As typical examples, we consider two cases: $s(x)=f(x)=1/x$ and $s(x)=1/x,f(x)=x$,
where $X$ is a three-dimensional Bessel process.
The first case is a standard example of a bubble:
the stock price $Y$ is the reciprocal of a three-dimensional Bessel process,
which is a strict local martingale.
An equivalent local martingale measure exists, namely $P_{x}^{s}$.
In the second case, $Y$ is the three-dimensional Bessel process $X$.
This is one of the best-known examples of markets without an equivalent local martingale measure; therefore, arbitrage opportunities exist.
These examples are discussed in Section \ref{sec:example}.

\citet{ruf2012} demonstrated that delta hedging provides the optimal trading strategy in terms of the minimal required initial capital to replicate a given terminal payoff in market models that include no equivalent local martingale measure and only a square-integrable market price of risk.
This method generalizes some results of the benchmark approach (\citet{platen2006benchmark}).
Furthermore, \citet{ruf2012} proposed a change of measure such that the dynamics of $Y$ becomes simpler.
This new measure coincides with the equivalent local martingale measure if it exists uniquely.

We adopt the method of \citet{ruf2012} to price a European option that pays $h(Y_{T})$ at time $T$,
where $h$ is a nonnegative Borel function.
In this case, the derivative price is defined as the minimal required initial capital to replicate a payoff.
Although arbitrage opportunities may exist for the investor,
an equivalent martingale measure is shown to exist in the hypothetical market for the fundraiser.
Hence, derivative pricing for the fundraiser falls into the standard theory because the stock price process (\ref{eq:sde_yj}) is regarded as the price of a stock that pays dividends $-2df(J_{t}^{*})$.

\subsection{Derivative pricing for investor}
According to Girsanov's theorem, a $P_{x}^{f}$-Brownian motion $\beta$
exists such that
the dynamics of $X,Y$, and $G$ are expressed as
\begin{eqnarray}
X_{t} &=& \beta_{t} - \int_{0}^{t} \frac{1}{2}T_{f}(X_{u}) du,\\
Y_{t} &=& Y_{0} + \int_{0}^{t} \sigma(Y_{u}) d\beta_{u}, \label{eq:stock} \\
G_{t} &=& 1 + \int_{0}^{t} G_{u} \theta(X_{u}) d\beta_{u}, \\
\beta_{t} &=& B_{t} + \int_{0}^{t} \theta(X_{u}) du \label{eq:bm_f}
\end{eqnarray}
for $t < T_{\infty}$, where $\sigma = f^{\prime} \circ f^{-1}$.
We extend $G$ such that $G_{t}=0$ for $t > T_{\infty}$,
because $X_{T_{\infty}}=0$ on $\{G_{T_{\infty}}=0\}$,
and define the function $H : [0,T] \times I^{\prime} \longrightarrow [0,\infty)$ as
\begin{eqnarray}
H(t,y) := P_{f^{-1}(y)}^{s}[h(f(X_{T-t}))Z_{T-t}]
= P_{f^{-1}(y)}^{f}[h(f(X_{T-t})):G_{T-t}>0]. \label{eq:stoch_solution}
\end{eqnarray}
If the payoff function $h$ is selected so that Assumption (A3) of \citet{ruf2012} is satisfied, which is the case when $h$ is of at most linear growth,
Theorems 4.1, 4.2, and 5.1 of \citet{ruf2012} can be applied to our market model.
\begin{proposition}
\label{prop:ruf}
For each $t,y\in [0,T) \times (0,\infty)$,
there exist some $C>0$ and some neighborhood $U$ of $(t,y)$ such that $H(\tau,\eta) < C$ for all $(\tau,\eta) \in U$.
Then, the time $t$ price $V_{t}$ for the investor of the European option that pays $h(Y_{T})$ at time $T$ is 
\begin{eqnarray}
V_{t} = H(t,f(X_{t})) = P_{X_{t}}^{s}\left[h(f(X_{T-t})) Z_{T-t}\right], \label{eq:investorprice}
\end{eqnarray}
$P_{x}^{s}$-a.s.,
and
\begin{eqnarray}
V_{t}  = P_{X_{t}}^{f}\left[h(f(X_{T-t})) :G_{T-t}>0\right], \label{eq:derivative}
\end{eqnarray}
$P_{x}^{f}$-a.s. (and thus, $P_{x}^{s}$-a.s.) on $\{G_{t} > 0\}$.
Furthermore, $H$ is a solution to 
\begin{eqnarray}
\left\{ \begin{array}{l}
v_{t}(t,y) = -\frac{1}{2} \sigma(y)^{2} v_{yy}(t,y) \label{eq:pde}, \\
v(T,y) = h(y).
\end{array} \right.
\end{eqnarray}
\end{proposition}

\subsection{Derivative pricing for fundraiser}
\label{sec:derivative_fundraiser}
Similarly, we introduce a probability measure under which $\tilde{G}$ is a local martingale.
We assume that $P_{x}^{s}[J_{t}=J_{t}^{*}]=1$ for all $t$, which is the case for $s(x)=1/x$.
We consider the filtered probability space $(\Omega,(\tilde{\mathcal{F}}_{t}^{0}),\tilde{P}_{x}^{1/s})$,
because $(X,J)$ in this probability space follows the same law 
as $(X,J)$ under $P_{x}^{s}$.
Suppose that $f$ and $1/s$ satisfy Assumption \ref{ass:equivalence}, which ensures that $\tilde{P}_{x}^{f}[T_{\infty}=\infty]=\tilde{P}_{x}^{1/s}[T_{\infty}=\infty]=1$.
Then, $\tilde{P}_{x}^{1/s}$ and $\tilde{P}_{x}^{f}$ are equivalent on $\tilde{\mathcal{F}}_{T}^{0}$
with $\tilde{P}_{x}^{f} = \tilde{Z}_{T}\tilde{P}_{x}^{1/s}$,
where $\tilde{Z}_{T}=1/\tilde{G}_{T}$.
The dynamics of $X,\tilde{Y}$, and $\tilde{G}$ in the filtered probability space $(\Omega,(\tilde{\mathcal{F}}_{t}^{f}),\tilde{P}_{x}^{f})$,
where $(\tilde{\mathcal{F}}_{t}^{f})$ is the $\tilde{P}_{x}^{f}$-augmentation of $(\tilde{\mathcal{F}}_{t}^{0})$,
are expressed as
\begin{eqnarray}
X_{t} &=& \tilde{\beta}_{t} - \int_{0}^{t} \frac{1}{2}T_{f}(X_{u}) du + 2J_{t}^{*}, \label{eq:x_pf_xj}\\
\tilde{Y}_{t} &=& Y_{0} + \int_{0}^{t} \tilde{Y}_{u} \frac{f^{\prime}}{f}(X_{u}) d\tilde{\beta}_{u}, \\
\tilde{G}_{t} &=& 1 + \int_{0}^{t} \tilde{G}_{u} \tilde{\theta}(X_{u}) d\tilde{\beta}_{u},
\end{eqnarray}
where $\tilde{\beta}$ is a $\tilde{P}_{x}^{f}$-Brownian motion with $\tilde{\beta}_{0} = x-2J_{0}^{*}$.
Let $\tilde{P}_{x,j}^{\cdot} := \tilde{P}_{x}^{\cdot}[\cdot | J_{0}^{*}=j]$
for $x,j$ with $x \ge j > 0$,
and
\begin{eqnarray}
\tilde{H}(t,y,k) &:=& \tilde{P}_{f^{-1}(y),f^{-1}(k)}^{1/s}[h(f(X_{T-t}))\tilde{Z}_{T-t}] \nonumber\\
&=& \tilde{P}_{f^{-1}(y),f^{-1}(k)}^{f}[h(f(X_{T-t}))]
\end{eqnarray}
for $(y,k) \in D := \{(y,k) \in I^{\prime} \times I^{\prime} \; : f^{-1}(y) \ge f^{-1}(k) > 0\}$.
The following assumption, for which a sufficient condition is provided later, is required:
\begin{assumption}
\label{ass:smooth}
For each $k>0$, $\tilde{H}(\cdot,\cdot,k): [0,T) \times I_{k}^{\prime} \longrightarrow \mathbb{R}$
is in $C^{1,2}$,
where $I_{k}^{\prime}=\{y \in I^{\prime} : f^{-1}(y) > f^{-1}(k) > 0\}$.
\end{assumption}

The following theorem is the main theorem of this study.
The derivative price for the fundraiser is derived
and it converges to the price for the investor as $j$ tends to $0$ under a certain condition, which can be used to compute derivative prices in markets with bubbles, as demonstrated in Section \ref{sec:bse}.
\begin{theorem}
\label{them:main}
Suppose
that $s$ satisfies Assumption \ref{ass:s},
$1/s$ satisfies Assumption \ref{ass:equivalence},
and $P_{x}^{s}[J_{t}=J_{t}^{*}]=1$ for all $t$;
that $f$ satisfies Assumptions \ref{ass:equivalence} and \ref{ass:f};
and that $\tilde{H}$ satisfies Assumption \ref{ass:smooth}.

Then, the time $t$ price $\tilde{V}_{t}$ for the fundraiser of the European option that pays $h(Y_{T})$ at time $T$ is 
\begin{eqnarray}
  \tilde{V}_{t} = \tilde{H}(t,f(X_{t}),f(J_{t}^{*}))
= \tilde{P}_{X_{t},J_{t}^{*}}^{f}\left[h(f(X_{T-t})) \right],
\end{eqnarray}
$\tilde{P}_{x,j}^{1/s}$-a.s.,
and $\tilde{H}$ is a solution to 
\begin{eqnarray}
\left\{ \begin{array}{l}
  v_{t}(t,y,k) = -\frac{1}{2} \sigma(y)^{2} v_{yy}(t,y,k), \\
  v(T,y,k) = h(y).
\end{array} \right.
\label{eq:pde_f}
\end{eqnarray}
Furthermore, suppose that $h$ is nondecreasing,
$f^{\prime}<0$, and $P_{x}^{f}[T_{\infty} = \infty]=1$.
Then, we obtain
\begin{eqnarray}
H(\tau,f(x))
= \lim_{j \downarrow 0} \tilde{H}(\tau,f(x),f(j)).
\end{eqnarray}
\end{theorem}

In the following, we prove Theorem \ref{them:main} using several lemmas,
derive an expression for $\tilde{H}$,
and provide a sufficient condition for Assumption \ref{ass:smooth}.
First, we present the derivative price and replication strategy.
\begin{lemma}
Suppose that $\tilde{H}$ satisfies Assumption \ref{ass:smooth}.
Then, for all $t \le T$,
\begin{eqnarray}
h(f(X_{T})) = \tilde{V}_{t} + \int_{t}^{T} \tilde{H}_{y}(u,f(X_{u}),f(J_{u}^{*})) \left(\frac{f(J_{u}^{*})}{f(J_{0}^{*})} \right)^{2}d\tilde{Y}_{u} \label{eq:replication},
\end{eqnarray}
$\tilde{P}_{x}^{1/s}$-a.s.,
and
$\tilde{H}$ satisfies the partial differential equation (\ref{eq:pde_f}).
\end{lemma}
\begin{proof}
Let $\tilde{M}_{t} = \tilde{H}(t,f(X_{t}),f(J_{t}^{*}))$.
Then, $\tilde{M}$ is a continuous $\tilde{P}_{x,j}^{f}$-martingale, because
\begin{eqnarray}
\tilde{M}_{t}
= \tilde{P}_{x,j}^{f}\left[h(f(X_{T})) \;\middle| \; \tilde{\mathcal{F}}_{t}^{f} \right]
\end{eqnarray}
for every $t \in [0,T]$.
Ito's formula leads to 
\begin{eqnarray}
\tilde{M}_{t}-\tilde{M}_{q}
&=&
\int_{q}^{t} \left(\tilde{H}_{t}(u,f(X_{u}),f(J_{u}^{*})) + \frac{1}{2} f^{\prime}(X_{u})^{2} \tilde{H}_{yy}(u,f(X_{u}),f(J_{u}^{*})) \right)du \nonumber\\
&+&
\int_{q}^{t} f^{\prime}(X_{u}) \tilde{H}_{y}(u,f(X_{u}),f(J_{u}^{*})) d\tilde{\beta}_{u}
\end{eqnarray}
for every $q \in \mathbb{Q}_{+}$
and $t \in [q,\theta_{q})$,
where $\theta_{q} = \inf \{u > q : J_{u}^{*} > J_{q}^{*} \}$.
Let $\Lambda$ be defined as
\begin{eqnarray}
\Lambda_{t} &:=& \tilde{M}_{t}-\tilde{M}_{0}
-
\int_{0}^{t} f^{\prime}(X_{u}) \tilde{H}_{y}(u,f(X_{u}),f(J_{u}^{*})) d\tilde{\beta}_{u} \nonumber\\
&-&
\int_{0}^{t} \left(\tilde{H}_{t}(u,f(X_{u}),f(J_{u}^{*})) + \frac{1}{2} f^{\prime}(X_{u})^{2} \tilde{H}_{yy}(u,f(X_{u}),f(J_{u}^{*})) \right)du.
\end{eqnarray}
Then,
$\Lambda$ is a continuous semimartingale that is constant on each interval on which $J^{*}$ is constant,
and
\begin{eqnarray}
&&\tilde{M}_{t}-\tilde{M}_{0}
-
\int_{0}^{t} f^{\prime}(X_{u}) \tilde{H}_{y}(u,f(X_{u}),f(J_{u}^{*})) d\tilde{\beta}_{u} \nonumber\\
&=&
\Lambda_{t}+\int_{0}^{t} \left(\tilde{H}_{t}(u,f(X_{u}),f(J_{u}^{*})) + \frac{1}{2} f^{\prime}(X_{u})^{2} \tilde{H}_{yy}(u,f(X_{u}),f(J_{u}^{*})) \right)du \label{eq:vanish}
\end{eqnarray}
is a local martingale with respect to the filtration generated by the Brownian motion $\tilde{\beta}$.
As $\{t \in [0,T] :dJ_{t}^{*}>0\}$ has a zero Lebesgue measure $P_{x}^{(1/2)}$-a.s., and thus, $\tilde{P}_{x}^{f}$-a.s,
(\ref{eq:vanish}) must vanish.
This leads to $\Lambda_{t} = 0$ and (\ref{eq:pde_f}).
\end{proof}
Note that, if $\tilde{H}$ is in $C^{1,2,1}([0,T)\times D)$,
Ito's formula can be used for the interval $[0,T]$
and
$\tilde{H}$ also satisfies $\tilde{H}_{k}(t,y,y) = -2 \tilde{H}_{y}(t,y,y)$.

We obtain two comparison results for the investor and fundraiser.
One concerns the sample paths and the other relates to expectations for a monotonic payoff function $h$.
For the latter result,
although we assume that $h$ is nondecreasing,
a similar result holds for a nonincreasing function $h$:
we obtain the opposite inequality of (\ref{eq:comp}) below 
under the assumption that $f^{\prime} > 0$ and $P_{x}^{f}[G_{\tau}>0]=1$,
or that $f^{\prime} < 0$, respectively.
\begin{lemma}
\label{lem:pathcomp}
Let $X^{\prime}$ be defined by a stochastic differential equation of the Skorokhod type:
\begin{eqnarray}
X_{t}^{\prime} = -2J_{0}^{\prime}+\beta_{t} - \int_{0}^{t} \frac{1}{2}T_{f}(X_{u}^{\prime}) du + 2J_{t}^{\prime},
\end{eqnarray}
where $\beta$ is defined by (\ref{eq:bm_f})
and $J_{0}^{\prime}$ is an arbitrary constant in $(0,x]$.
Subsequently, we obtain $X_{t} \le X_{t}^{\prime}$ for all $t \ge 0$, almost surely with respect to $P_{x}^{f}$.
\end{lemma}
\begin{proof}
Let $\Delta_{t}^{+} = (X_{t}-X_{t}^{\prime})_{+}$.
Then, we obtain
\begin{eqnarray}
\Delta_{t}^{+}
&=& -\frac{1}{2} \int_{0}^{t} 1_{\{X_{u} \ge X_{u}^{\prime}\}} \left(T_{f}(X_{u})-T_{f}(X_{u}^{\prime}) \right) du - 2\int_{0}^{t} 1_{\{X_{u} \ge X_{u}^{\prime}\}} dJ_{u}^{\prime} \nonumber\\
&\le& \frac{1}{2}\left(\sup_{\xi \in K} |T_{f}^{\prime}(\xi)| \right) \int_{0}^{t} 1_{\{X_{u} \ge X_{u}^{\prime}\}} |X_{u}-X_{u}^{\prime}| du,
\end{eqnarray}
where $K$ denotes the interval $K = [J_{0}^{\prime},\sup_{u \le t} X_{u}]$.
The conclusion follows from Gronwall's inequality.
\end{proof}

\begin{lemma}
\label{lem:comp}
Suppose that $h$ is nondecreasing.
Then,
if $f^{\prime} < 0$ and $P_{x}^{f}[G_{\tau}>0]=1$, we obtain
\begin{eqnarray}
  P_{x}^{f}\left[h(f(X_{\tau})) : G_{\tau}>0 \right]
  \ge
 \tilde{P}_{x,j}^{f}\left[h(f(X_{\tau})) \right]\label{eq:comp}
\end{eqnarray}
for all $x,j$ with $x \ge j>0$,
and if $f^{\prime} > 0$, we obtain the opposite inequality to (\ref{eq:comp}).
\end{lemma}
\begin{proof}
Let $X^{\prime}$ be as in Lemma \ref{lem:pathcomp} with $J_{0}^{\prime} = j$.
Subsequently, the law of $X^{\prime}$ under $P_{x}^{f}$ is the same as that of $X$ under $\tilde{P}_{x,j}^{f}$.
If $f^{\prime} < 0$, we obtain $f(X_{\tau}^{\prime}) \le f(X_{\tau})$.
Owing to $P_{x}^{f}[G_{\tau}>0]=1$,
we obtain
\begin{eqnarray}
  P_{x}^{f}\left[h(f(X_{\tau})) : G_{\tau}>0 \right]
  =
  P_{x}^{f}\left[h(f(X_{\tau})) \right]
  \ge
  P_{x}^{f}\left[h(f(X_{\tau}^{\prime})) \right],
\end{eqnarray}
which is equal to the right-hand side of (\ref{eq:comp}).
Similarly, if $f^{\prime} > 0$, we obtain $f(X_{\tau}^{\prime}) \ge f(X_{\tau})$
and
\begin{eqnarray}
  P_{x}^{f}\left[h(f(X_{\tau})) : G_{\tau}>0 \right]
  \le
  P_{x}^{f}\left[h(f(X_{\tau})) \right]
  \le
  P_{x,j}^{f}\left[h(f(X_{\tau}^{\prime})) \right],
\end{eqnarray}
which is equal to the right-hand side of (\ref{eq:comp}).
\end{proof}

Counter-intuitively, 
the delta,
which is the derivative of $\tilde{H}$ with respect to $y$ multiplied by $(f(J^{*})/f(J_{0}^{*}))^{2}$,
is not necessarily positive, even if the payoff $h$ increases.
This is partly because the fundraiser knows that $y=f(j)$ is a reflecting boundary.
Consequently, the pathwise comparison result for $X$ that is defined by (\ref{eq:x_pf_xj}) with a different $x$ does not hold.
See Section \ref{sec:-1/2} for further details.

We show that the expectation of $h(f(X_{\tau}))$ on $\{\tau < \tau_{j}\}$
with respect to $\tilde{P}_{x,j}^{f}$ converges to $P_{x}^{f}[h(f(X_{\tau}))]$.
\begin{lemma}
\label{lem:mainpart}
For $x,j$ with $x \ge j > 0$, we obtain
\begin{eqnarray}
\tilde{P}_{x,j}^{f}[h(f(X_{\tau}))]
\ge W_{x}[h(f(X_{\tau})) \mathcal{S}_{\tau}^{f}(X) : \tau < \tau_{j} ], \label{eq:mainpart}
\end{eqnarray}
where $\tau_{j} = \inf \{t > 0 \; : X_{t} = j \}$.
\end{lemma}
\begin{proof}
  The equivalence between $\tilde{P}_{x}^{f}$ and $P_{x}^{(1/2)}$ on $\tilde{\mathcal{F}}_{\tau}^{0}$
  and $J_{t} = J_{0}$ for $t < \tau_{j}$
  lead to
\begin{eqnarray}
\tilde{P}_{x,j}^{f}[h(f(X_{\tau}))]
&=&
P_{x}^{(1/2)}[h(f(X_{\tau})) \tilde{\mathcal{S}}_{\tau}^{f}(X) | J_{0}=j ] \nonumber\\
&\ge&
P_{x}^{(1/2)}[h(f(X_{\tau})) \mathcal{S}_{\tau}^{f}(X) 1_{\{\tau < \tau_{j}\}} | J_{0}=j ] \nonumber\\
&=&
W_{x}[h(f(X_{\tau})) \mathcal{S}_{\tau}^{f}(X) : \tau < \tau_{j}].
\end{eqnarray}
\end{proof}

Theorem \ref{them:main} is completed when the following lemma has been established:
\begin{lemma}
\label{lem:convergence}
Suppose that $h$ is nondecreasing,
$f^{\prime}<0$, and $P_{x}^{f}[T_{\infty} = \infty]=1$.
Then, we obtain
\begin{eqnarray}
P_{x}^{f} [h(f(X_{\tau}))] = \lim_{j \downarrow 0} \tilde{P}_{x,j}^{f} [h(f(X_{\tau}))].
\end{eqnarray}
\end{lemma}
\begin{proof}
According to Lemmas \ref{lem:comp} and \ref{lem:mainpart}, 
for all $j>0$, we obtain
\begin{eqnarray}
P_{x}^{f}[h(f(X_{\tau}))]
&\ge&
\tilde{P}_{x,j}^{f}[h(f(X_{\tau}))] \nonumber \\
&\ge&
W_{x}[h(f(X_{\tau})) \mathcal{S}_{\tau}^{f}(X) : \tau < \tau_{j} ].
\end{eqnarray}
As $j \downarrow 0$, the final quantity converges to
\begin{eqnarray}
W_{x}[h(f(X_{\tau})) \mathcal{S}_{\tau}^{f}(X) : \tau < \tau_{0} ]
&=&
\lim_{n \rightarrow \infty} W_{x}[h(f(X_{\tau})) \mathcal{S}_{\tau}^{f}(X) : \tau  < T_{n} ] \nonumber \\
&=&
\lim_{n \rightarrow \infty}  P_{x}^{f}[h(f(X_{\tau})) : \tau < T_{n}] \nonumber \\
&=&
P_{x}^{f}[h(f(X_{\tau}))].
\end{eqnarray}
\end{proof}

We derive the boundary condition at $(T-\tau,f(j))$ that $\tilde{H}(\cdot,\cdot,f(j))$ should satisfy; that is, $\Theta_{j}^{h}(\tau)$ below.
We remark that the stochastic differential equation (\ref{eq:singular}) with a singular drift coefficient in the following lemma has a nonnegative unique solution (Theorem 1.1 of \citet{SaishoTanemura1990}).
\begin{lemma}
For each $j>0$,
let $P_{0}^{1/\tilde{f}_{j}}$ be the law of
the diffusion $X^{(j)}$ defined by
\begin{eqnarray}
X_{t}^{(j)} = \gamma_{t} - \int_{0}^{t} \frac{1}{2}T_{1/\tilde{f}_{j}}(X_{u}^{(j)})du, \label{eq:singular}
\end{eqnarray}
where $\gamma$ is a Brownian motion with $\gamma_{0}=0$ and $\tilde{f}_{j}=f_{j}(0)-f_{j}$.
Then, we obtain
\begin{eqnarray}
\Theta_{j}^{h}(\tau)
:= P_{0}^{1/\tilde{f}_{j}}[h(f_{j}(X_{\tau}))]
= \int_{0}^{\infty} \frac{q_{j}^{h}(x)}{x} \theta_{j}(\tau,x) p^{(1/2)}(\tau,x) dx, \label{eq:boundary_value}
\end{eqnarray}
where $p^{(1/2)}(\tau,x) := \frac{2x^{2}}{\sqrt{2\pi \tau^{3}}} \mathrm{e}^{-\frac{x^{2}}{2\tau}}$ is the density of the three-dimensional Bessel process that starts at $0$,
and
\begin{eqnarray}
q_{j}^{h}(x)
&=& 
\frac{h(f_{j}(x))}{f_{j}(x)} \left(1-\frac{f_{j}(x)}{f_{j}(0)}\right) \frac{1}{\sqrt{(1/f_{j})^{\prime}(0)(1/f_{j})^{\prime}(x)}},\\
\theta_{j}(\tau,x) &=& P_{0}^{(1/2)}\left[\mathrm{e}^{\frac{1}{4}\int_{0}^{\tau}S_{f_{j}}(X_{u})du} \;\middle| \; X_{\tau}=x \right].
\end{eqnarray}
\end{lemma}
\begin{proof}
For each $j>0$, $\chi_{\cdot}:=X_{\cdot + \tau_{j}}-j$ satisfies
\begin{eqnarray}
\chi_{t} = \tilde{\beta}_{t} - \int_{0}^{t} \frac{1}{2}T_{f}(\chi_{u}+j)du+2(J_{t+\tau_{j}}^{*}-j).
\end{eqnarray}
Theorem 1.1 of \citet{SaishoTanemura1990} shows 
that $\chi$ is the unique nonnegative solution to (\ref{eq:singular})
for some Brownian motion $\gamma$ with $\gamma_{0}=0$, independent of $\tilde{\mathcal F}_{\tau_{j}}^{0}$.
Then, we obtain
\begin{eqnarray}
\tilde{P}_{x,j}^{f}[h(f(X_{\tau_{j}+\tau})) | \tilde{\mathcal F}_{\tau_{j}}^{0}]
= P_{0}^{1/\tilde{f}_{j}}[h(f_{j}(X_{\tau}))].
\end{eqnarray}

For any Borel set $\Gamma \subseteq \{\tau_{j} < \tau \}$ of $\tilde{\mathcal F}_{\tau_{j}}^{0}$,
the path decomposition theorem of the three-dimensional Bessel process according to Williams leads to
\begin{eqnarray}
\tilde{P}_{x,j}^{f}[h(f(X_{\tau})) : \Gamma]
&=&
P_{x}^{(1/2)}[h(f(X_{\tau}))\tilde{\mathcal{S}}_{\tau}^{f}(X)1_{\Gamma} | J_{0}^{*} = j]\nonumber\\
&=&
W_{x}\left[\left. P_{0}^{(1/2)}[h(f(X_{u}+j))\tilde{\mathcal{S}}_{u}^{f}(X+j)] \right|_{u = \tau-\tau_{j}} \mathcal{S}_{\tau_{j}}^{f}(X)  : \Gamma \right]\nonumber\\
&=&
\tilde{P}_{x,j}^{f}\left[\left. P_{0}^{(1/2)}[h(f(X_{u}+j))\tilde{\mathcal{S}}_{u}^{f}(X+j)] \right|_{u = \tau-\tau_{j}} : \Gamma \right].
\end{eqnarray}
Then, using the fact that $J_{\tau}$ is uniformly distributed on $[0,X_{\tau}]$ under $P_{x}^{(1/2)}$, we obtain the following on $\{\tau_{j} < \tau \}$:
\begin{eqnarray}
\tilde{P}_{x,j}^{f}[h(f(X_{\tau_{j}+\tau})) | \tilde{\mathcal F}_{\tau_{j}}^{0}]
= P_{0}^{(1/2)}[h(f_{j}(X_{\tau}))\tilde{\mathcal{S}}_{\tau}^{f_{j}}(X)]
= \int_{0}^{\infty} \frac{q_{j}^{h}(x)}{x} \theta_{j}(\tau,x) p^{(1/2)}(\tau,x) dx.
\end{eqnarray}
\end{proof}

Finally, we provide a sufficient condition for Assumption \ref{ass:smooth}.
\begin{lemma}
\label{lem:smooth}
Let $h$ be a nonnegative continuous function.
Then, Assumption \ref{ass:smooth} is satisfied
if $\tilde{H}(\cdot,f(x), f(j))$ is continuous and
$\lim_{t \rightarrow T}\tilde{H}(t,f(x), f(j)) = h(f(x))$
for $x \ge j >0$,
which is the case if
\begin{itemize}
\item[(1)] $f^{\prime} < 0$, or
\item[(2)] $h$ is of at most linear growth, $1/f$ satisfies Assumption \ref{ass:equivalence},
and $\tilde{P}_{x,j}^{1/f}[f(J_{\tau})^{2}] < \infty$ holds for all $\tau \le T$.
\end{itemize}
\end{lemma}
\begin{proof}
For a fixed $j > 0$, let $x_{n}$ be a sequence with $j < x_{n} \nearrow \infty$
and $w$ be a solution to the following initial-boundary value problem for each $n$:
\begin{eqnarray}
\left\{ \begin{array}{l}
\frac{1}{2}w_{xx}(t,x) - \frac{1}{2} T_{f}(x) w_{x}(t,x) + w_{t}(t,x) = 0, \\
w(T,x) = h(f(x)),\\
w(t,j) = \Theta(T-t,j),\\
w(t,x_{n}) = g(T-t,x_{n}),
\end{array} \right.
\end{eqnarray}
where $x \in (j,x_{n}), 0 < t \le T$,
$g:=g_{1}+g_{2}$, and
\begin{eqnarray}
g_{1}(\tau,x) &:=& P_{x}^{f} \left[h(f(X_{\tau})) : \tau < \tau_{j} \right], \\
g_{2}(\tau,x) &:=& P_{x}^{f} \left[\Theta(\tau-\tau_{j},j) : \tau > \tau_{j} \right].
\end{eqnarray}
Subsequently, because the initial-boundary value is continuous,
according to Theorem 5.2 in Chapter 6 of \citet{Friedman2012},
the solution is unique
and has a stochastic representation: 
\begin{eqnarray}
w(t,x)
&=&
P_{x}^{f} \left[h(f(X_{T-t})) : T-t < \tau_{j}^{n} \right]
+ P_{x}^{f} \left[\Theta((T-t)-\tau_{j},j) : \tau_{j} = \tau_{j}^{n} < T-t \right]\nonumber\\
&&+ P_{x}^{f} \left[g((T-t)-\tau_{j}^{n},X_{\tau_{j}^{n}}) : \tau_{j} \neq \tau_{j}^{n} < T-t \right],
\end{eqnarray}
where $\tau_{j}^{n} = \inf\{t : X_{t} = x_{n}\} \wedge \tau_{j}$
if $g$ and $\Theta$ are continuous in $t \in [0,T]$,
which is the case because $g(T-t,x)=\tilde{H}(t,f(x),f(j))$ and $\Theta(T-t,j)=\tilde{H}(t,f(j),f(j))$.
Straightforward computations show that $w(t,x)=\tilde{H}(t,f(x),f(j))$ is in $C^{1,2}([0,T) \times (j,\infty))$.

In the case of (1), $\tilde{P}_{x,j}^{f}[h(f(X_{\tau}))]$
is continuous in $\tau$ because $f(X)$ is bounded by $f(j)$ and $h$ is of at most linear growth.
In the case of (2), Assumption \ref{ass:equivalence} ensures that 
$\tilde{P}_{x}^{f}$ and $\tilde{P}_{x}^{1/f}$ are equivalent on $\tilde{\mathcal{F}}_{\tau}^{0}$.
Then, 
\begin{eqnarray}
\tilde{P}_{x,j}^{f}[h(f(X_{\tau}))]
= 
\tilde{P}_{x,j}^{1/f} \left[h(f(X_{\tau}))\frac{f(x)}{f(X_{\tau})}\frac{f(J_{\tau})^{2}}{f(j)^{2}} \right]
\end{eqnarray}
is continuous in $\tau$ owing to the dominated convergence theorem
with the integrability of $f(J_{\tau})^{2}$.
\end{proof}

\section{Black--Scholes equation in the presence of a bubble}
\label{sec:bse}
In this section, we consider derivative prices in the presence of a bubble. In this case,
a bubble is a discounted price process that is a nonnegative strict local martingale; hence, a supermartingale, under a risk-neutral measure.
The underlying process $Y$ that satisfies (\ref{eq:stock}) is a strict local martingale if and only if 
\begin{eqnarray}
\int_{1}^{\infty} \frac{y}{\sigma(y)^{2}} dy < \infty, \label{eq:bubble}
\end{eqnarray}
according to Theorem 1 of \citet{Kotani2006}.

In the presence of a bubble, the Black--Scholes equation (\ref{eq:pde}) admits multiple solutions for $h$ of linear growth,
among which the stochastic solution $v^{h}(\tau,y):=H(T-\tau,y)$ is characterized as the smallest solution.
A typical and an important example is the case of the forward contract; that is, $h(y)=y$,
in which case we express the stochastic solution using $v^{*}$ instead of $v^{h}$
for special emphasis.
The stochastic solutions $v^{*}$ and $v(\tau,y)=y$
are two distinct solutions owing to the supermartingale property of $Y$:
\begin{eqnarray}
v^{*}(\tau,y) = P_{f^{-1}(y)}^{f}[f(X_{\tau})] < y = v(\tau,y).
\end{eqnarray}

Several attempts have been made in previous research to recover the uniqueness and calculate the stochastic solution numerically (see \citet{Tysk2009}, \citet{EkstromLotstedtSydowTysk2011}, \citet{SongYang2015}, \citet{Cetin2018}, and \citet{tsuzuki2024boundary}).
We add a new numerical procedure to this list:
the derivative price is approximated from below by those for a fundraiser if $h$ is nondecreasing (Theorem \ref{them:main}).
See \citet{tsuzuki2024boundary} for numerical experiments of these procedures,
which showed that the method proposed by \citet{tsuzuki2024boundary} outperformed the others.

Hereafter in this section, we assume Assumption \ref{ass:boundary} below,
under which $0$ is an entrance boundary for the diffusion $X$ under $P_{\cdot}^{f}$
according to the conditions $f(0)=\infty$ and (\ref{eq:entrance}),
and $\infty$ is not according to (1) of Assumption \ref{ass:f}.
Then, we can extend $X$ to a continuous Feller process on $[0,\infty)$ using Theorem 33.13 of \citet{Kallenberg2021},
and define $P_{0}^{f}$ as the law of $X$ with $X_{0}=0$.
In particular, $P_{0}^{f}[X_{t}>0]=1$ for $t > 0$ holds.
We remark that the condition (\ref{eq:bubble}) on $\sigma$ for the presence of a bubble
is equivalent to the condition (\ref{eq:entrance}) on $f$.
\begin{assumption}
\label{ass:boundary}
Suppose that $f$ satisfies (1) of Assumption \ref{ass:f}, $f(0) = \infty$ and,
\begin{eqnarray}
\int_{0}^{1} f(\xi)m_{f}(d\xi) < \infty, \label{eq:entrance}
\end{eqnarray}
where $m_{f}(dx):=\frac{2}{|f^{\prime}(x)|}dx$.
\end{assumption}

We show that the boundary condition $\Theta_{j}^{h}$ also converges to the stochastic solution with the underlying price infinity, namely $\Theta_{0}^{h}(\tau) := P_{0}^{f}[h(f(X_{\tau}))]=\lim_{j\downarrow 0} \Theta_{j}^{h}(\tau)$,
derive expressions for $\Theta_{j}^{h}$,
and obtain the boundary condition of \citet{tsuzuki2024boundary}
\begin{eqnarray}
v^{h}(\tau,\infty)
= -2 \int_{0}^{\infty} v_{\tau}^{*}(\tau,y) \frac{h(y)}{\sigma(y)^{2}} dy
= -\int_{0}^{\infty} h(y)v_{yy}^{*}(\tau,y)dy, \label{eq:boundary}
\end{eqnarray}
which ensures that the Black--Scholes equation (\ref{eq:pde}) has a unique solution.
The main theorem in this section is as follows:
\begin{theorem}
\label{thm:boundary}
Under Assumption \ref{ass:boundary},
for a continuous, nonnegative function $h$ and $\tau >0$, we obtain
\begin{eqnarray}
\Theta_{j}^{h}(\tau) &=& \int_{j}^{\infty} \left(1-\frac{f(x)}{f(j)}\right) \frac{d}{d \tau} P_{x}^{1/f}[\tau_{j} \le \tau] \frac{h(f(x))}{f(x)} m_{1/f}(dx), \label{eq:Theta_j}\\
\Theta_{0}^{h}(\tau) &=& \int_{0}^{\infty}\frac{d}{d \tau} P_{x}^{1/f}[\tau_{0} \le \tau] \frac{h(f(x))}{f(x)} m_{1/f}(dx).\label{eq:Theta_0}
\end{eqnarray}
Furthermore, we obtain $\Theta_{j}^{h}(\tau) \nearrow \Theta_{0}^{h}(\tau) = v^{h}(\tau,\infty)$ as $j\downarrow 0$ for $\tau >0$ if $h$ is nondecreasing.
\end{theorem}

We remark that $P_{x}^{1/f}[\tau_{j} > \tau]$ for each $j \ge 0$, which is the complementary cumulative distribution function of the first hitting time for $X$ to reach $j$ with respect to $P_{x}^{1/f}$,
is the solution to
\begin{eqnarray}
\left\{ \begin{array}{l}
u_{\tau} = \frac{1}{2}u_{xx}-\frac{1}{2}T_{1/f} u_{x}, \\
u(0,x) = 1,\\
u(\tau,j)=0,\label{eq:comp_distribution}
\end{array} \right.
\end{eqnarray}
for $(\tau,x) \in (0,\infty) \times (j,\infty)$.
See Sections 4 and 5 of \citet{KaratzasRuf2013} for the analytical properties of the solutions.
In particular, the first hitting time $\tau_{j}$ has a density function
that is continuous in $(\tau,x)$.

Once (\ref{eq:Theta_j}) and (\ref{eq:Theta_0}) have been established,
we can obtain (\ref{eq:boundary}) by 
\begin{eqnarray}
v^{*}(\tau,f(x))
= P_{x}^{f} \left[f(X_{\tau})\right]
= f(x)P_{x}^{1/f}[\tau < \tau_{0}],
\end{eqnarray}
and
we obtain the expression (\ref{eq:koexpression}) of $\Theta_{j}^{h}$ below in terms of the price $v^{o,j}$ of a knock-out forward contract
with barrier level $f(j)$:
\begin{eqnarray}
v^{o,j}(\tau,y) := P_{f^{-1}(y)}^{f} [f(X_{\tau}):\tau < \tau_{j}]
= y P_{f^{-1}(y)}^{1/f} [\tau < \tau_{j}]  .\label{eq:ko}
\end{eqnarray}
We also remark that $v^{o,j}$ is decreasing in $\tau$, and hence, $v_{\tau}^{o,j} \le 0$.
\begin{proposition}
\label{prop:average_ko}
Let $\tau > 0$ and $\lambda_{j}(y):=\left(1-\frac{y}{f(j)}\right)h(y)$ for $j>0$.
If $\lambda_{j}$ is in $C^{2}([0,f(j)])$
and satisfies the following vanishing conditions at $y=0,f(j)$:
\begin{eqnarray}
\lambda_{j}(y)v_{y}^{o,j}(\tau,y)=\lambda_{j}^{\prime}(y)v^{o,j}(\tau,y)=0,
\end{eqnarray}
we obtain 
\begin{eqnarray}
\Theta_{j}^{h}(\tau) = - \int_{0}^{f(j)} \lambda_{j}^{\prime\prime}(y) v^{o,j}(\tau,y) dy.\label{eq:koexpression}
\end{eqnarray}
\end{proposition}
\begin{proof}
According to (\ref{eq:Theta_j}) and (\ref{eq:ko}), we obtain
\begin{eqnarray}
\Theta_{j}^{h}(\tau)
= -\int_{j}^{\infty} \lambda_{j}(f(x)) v_{\tau}^{o,j}(\tau,f(x)) m_{f}(dx)
=-\int_{0}^{f(j)} \lambda_{j}(y) v_{\tau}^{o,j}(\tau,y) \frac{2dy}{\sigma^{2}(y)}
\end{eqnarray}
by changing the variable $y=f(x)$.
For $\lambda_{j} \in C^{2}([0,f(j)])$ with the vanishing conditions, we obtain 
\begin{eqnarray}
\Theta_{j}^{h}(\tau)
= -\int_{0}^{f(j)} \lambda_{j}(y) v_{yy}^{o,j}(\tau,y)  dy
= -\int_{0}^{f(j)} \lambda_{j}^{\prime\prime}(y) v^{o,j}(\tau,y) dy.
\end{eqnarray}
\end{proof}  

\subsection{Convergence}
In this section, we show that $\Theta_{0}^{h}(\tau)=\lim_{j \downarrow 0}\Theta_{j}^{h}(\tau)$ for each $\tau>0$ if $h$ is nondecreasing.
As mentioned previously, 
$\{P_{x}^{f}\}_{x\ge 0}$ is the family of the law of a continuous Feller process on $[0,\infty)$,
which leads to the convergence of laws $P_{x}^{f}$ to $P_{0}^{f}$ as $x \downarrow 0$.
This convergence can be realized in the pathwise sense.
\begin{lemma}
\label{lem:feller}
For $x > 0$,
let $\xi^{(x)}$ be the stochastic process on the probability space $(\Omega,\mathcal{F},W_{0})$ that is defined by the stochastic differential equation
\begin{eqnarray}
\xi_{t}^{(x)} = x+X_{t} -\frac{1}{2} \int_{0}^{t}T_{f}(\xi_{u}^{(x)}) du.\label{eq:transient}
\end{eqnarray}
Then, $\xi^{(x)}$ converges decreasingly and uniformly in any compact subset in
$[0,\infty)$ as $x \downarrow 0$, $W_{0}$ almost surely.
\end{lemma}
\begin{proof}
For each $\omega \in \Omega$ and $t \ge 0$, 
we define $\xi_{t}(\omega)$ as $\xi_{t}^{(x)}(\omega) \downarrow \xi_{t}(\omega)$ as $x\downarrow 0$.
For each $\delta >0$, the random variable $\xi_{\delta}^{(x)}$ converges to $\xi_{\delta}$, the law of which is $P_{0}^{f}\circ X_{\delta}^{-1}$.
In particular, $W_{0}[\xi_{\delta}>0]=P_{0}^{f}[X_{\delta}>0]=1$ holds.
For a random variable $a > 0$, we can define $\eta^{(a)}$ as the solution to
\begin{eqnarray}
\eta_{t}^{(a)} = a+(X_{t+\delta}-X_{\delta}) -\frac{1}{2} \int_{0}^{t}T_{f}(\eta_{u}^{(a)}) du.
\end{eqnarray}
According to the comparison theorem, we obtain
\begin{eqnarray}
W_{0} \left[\xi_{t+\delta}^{(x)} = \eta_{t}^{(\xi_{\delta}^{(x)})} \ge \eta_{t}^{(\xi_{\delta})}, t > 0 \right]=1
\end{eqnarray}
for each $x>0$.
Because  
$W_{0} \left[\inf_{t \le T} \eta_{t}^{(\xi_{\delta})} > 0 \right]=1$ holds for each $T>0$,
and $T_{f}$ is locally Lipschitz in $(0,\infty)$,
the convergence of $\xi_{t+\delta}^{(x)}$ as $x \downarrow 0$ is uniform in $t \in [0,T]$, $W_{0}$ almost surely.
This implies that $\xi$ is continuous in $(0,\infty)$.
The continuity of $\xi$ at $t=0$ follows from $\xi_{0+} \le \lim_{t\downarrow 0} \xi_{t}^{(x)}=x$ for any $x>0$.
\end{proof}
We remark that the stochastic differential equation (\ref{eq:transient}) with $x=0$ is singular at $x=0$
and does not necessarily have a strong solution (see Section 2.4 of \citet{ChernyEngelbert2005}).

\begin{lemma}
\label{lem:tightness}
The laws $P_{0}^{1/\tilde{f}_{j}}$ converge to $P_{0}^{f}$ as $j \downarrow 0$.
\end{lemma}
\begin{proof}
We use the same notation as that in the proof of Lemma \ref{lem:feller},
and define $X^{(j)}, j >0$ as the strong solution to (\ref{eq:singular}), with $\gamma$ replaced with $X$.
Let $(\mu_{j})_{j > 0}$ be an increasing sequence of functions that are defined as
\begin{eqnarray}
\mu_{j}(x)
:= \left\{ \begin{array}{ll}
\infty, & x \le j \\
-\frac{1}{2}T_{f}(x) + \frac{-f^{\prime}(x)}{f(j)-f(x)}, & x > j,
\end{array} \right.
\end{eqnarray}
and $\chi^{(j)} :=X^{(j)}+j$ for $j >0$.
Then, $\chi^{(j)}$ satisfies
\begin{eqnarray}
\chi_{t}^{(j)} = j+X_{t} + \int_{0}^{t} \mu_{j}(\chi_{u}^{(j)}) du. \label{eq:sde_xi}
\end{eqnarray}
We define $\chi$ as $\chi_{t}^{(j)}(\omega) \downarrow \chi_{t}(\omega)$ as $j \downarrow 0$.

When comparing $\chi_{0}^{(j)}$ and $\xi_{0}^{(x)}$, $\mu_{j}$ and $-\frac{1}{2}T_{f}$,
the comparison theorem implies that $\chi_{t}^{(j)} \ge \xi_{t}$,
and subsequently, that $W_{0}[\chi_{t} \ge \xi_{t}, t \ge 0]=1$.
Hence, $\inf_{j>0,t\in [\delta,1/\delta]} \chi_{t}^{(j)}(\omega) > 0$ for $\delta \in (0,1)$
implies that $\chi$ is a continuous process that satisfies 
\begin{eqnarray}
\chi_{t}-\chi_{\delta} = X_{t}-X_{\delta} -\frac{1}{2} \int_{\delta}^{t} T_{f}(\chi_{u}) du,\; t \ge \delta.
\end{eqnarray}
It follows for $x > 0$ that 
$\chi_{t} \le \xi_{t}^{(x)}$ for $t \ge \delta$ on $\{\chi_{\delta} \le \xi_{\delta}^{(x)}\}$.
By letting $\delta \downarrow 0$, and then letting $x \downarrow 0$, we conclude that $W_{0}[\xi_{t} \ge \chi_{t}, t \ge 0]=1$.
Dini's theorem leads to the compact uniform convergence of $\chi^{(j)}$ to $\xi$.
\end{proof}

\begin{proposition}
\label{prop:convergence}
For each $\tau>0$, the sequence $\Theta_{j}^{h}(\tau)$ increasingly converges to $\Theta_{0}^{h}(\tau)$ as $j\downarrow 0$ if $h$ is nondecreasing.
\end{proposition}
\begin{proof}
We use the same notation as that in the proof of Lemma \ref{lem:tightness}.
From (\ref{eq:boundary_value}), we obtain
\begin{eqnarray}
\Theta_{j}^{h}(\tau) = P_{0}^{\tilde{f}_{j}}[h(f_{j}(X_{\tau}))]
= W_{0}[h(f_{j}(X_{\tau}^{(j)}))]
= W_{0}[h(f(\chi_{\tau}^{(j)}))],
\end{eqnarray}
which shows the monotonicity of $\Theta_{j}^{h}(\tau)$ because $(\chi_{\tau}^{(j)})_{j \ge 0}$ is a nondecreasing sequence and $h\circ f$ is nonincreasing.
We note the following for the convergence:
\begin{eqnarray}
W_{0}\left[h(f(\chi_{\tau}^{(j)})) \wedge K \right]
\le \Theta_{j}^{h}(\tau)
\le W_{0}[h(f(\chi_{\tau}))]
= P_{0}^{f}[h(f(X_{\tau}))]
\end{eqnarray}
holds for all $K>0$.
The proof is completed by letting $j\downarrow 0$, and then letting $K\uparrow \infty$.
\end{proof}

\subsection{Expression of $\Theta_{j}^{h}$}
Using the expression of (\ref{eq:boundary_value}),
we establish Lemma \ref{lem:key} below,
and derive (\ref{eq:Theta_j}) and (\ref{eq:Theta_0}).
\begin{lemma}
\label{lem:key}
For $k,\lambda \in C_{K}((0,\infty))$, we obtain
\begin{eqnarray}
\int_{0}^{\infty} k(\tau) P_{0}^{1/\tilde{f}_{j}}\left[\lambda(X_{\tau}) h(f_{j}(X_{\tau})) \right] d\tau
= \int_{0}^{\infty} \left(1-\frac{f_{j}(x)}{f_{j}(0)}\right) \lambda(x) 
P_{x}^{1/f_{j}}\left[k(\tau_{0})\right] \frac{h(f_{j}(x))}{f_{j}(x)} m_{1/f_{j}}(dx).\label{eq:key}
\end{eqnarray}
\end{lemma}
\begin{proof}
We define $\Lambda_{j}(\tau,\lambda), K_{j}(x,k)$, and $\Psi_{j}(k,\lambda)$ as follows:
\begin{eqnarray}
\Lambda_{j}(\tau,\lambda)
&:=& \int_{0}^{\infty} \lambda(x) \frac{q_{j}^{h}(x)}{x}\theta_{j}(\tau,x) p^{(1/2)}(\tau,x)dx,\\
K_{j}(x,k) &:=& 2\int_{0}^{\infty} k(\tau) \theta_{j}(\tau,x) \frac{p^{(1/2)}(\tau,x)}{2x} d\tau, \\
\Psi_{j}(k,\lambda)
&:=& \int_{0}^{\infty}\int_{0}^{\infty} k(\tau)\lambda(x) \frac{q_{j}^{h}(x)}{x} \theta_{j}(\tau,x) p^{(1/2)}(\tau,x) dxd\tau.
\end{eqnarray}
Then, we obtain $\Lambda_{j}(\tau,\lambda)=P_{0}^{1/\tilde{f}_{j}}\left[\lambda(X_{\tau}) h(f_{j}(X_{\tau})) \right]$ using (\ref{eq:boundary_value}),
and the left-hand side of (\ref{eq:key}) is equal to
\begin{eqnarray}
\Psi_{j}(k,\lambda)
= \int_{0}^{\infty} \Lambda_{j}(\tau,\lambda) k(\tau) d\tau
=\int_{0}^{\infty} \lambda(x)q_{j}^{h}(x) K_{j}(x,k) dx.
\end{eqnarray}
The time reversal of Williams (Corollary 4.6 of VII and Exercise 3.12 of XI in \citet{revuzyor}) yields
\begin{eqnarray}
\theta_{j}(\tau,x) = P_{0}^{(1/2)}\left[\mathrm{e}^{\frac{1}{4}\int_{0}^{\tau}S_{f_{j}}(X_{u})du} \;\middle| \; X_{\tau}=x \right]
= W_{x}\left[\mathrm{e}^{\frac{1}{4}\int_{0}^{\tau}S_{f_{j}}(X_{u})du} \;\middle| \; \tau_{0}=\tau \right],
\end{eqnarray}
and the fact that the density of $\tau_{0}$ under $W_{x}$ is $p^{(1/2)}(\cdot,x)/2x$
yields
\begin{eqnarray}
K_{j}(x,k)
= 2\sqrt{\frac{(1/f_{j})^{\prime}(0)}{(1/f_{j})^{\prime}(x)}}W_{x}\left[k(\tau_{0}) \mathcal{S}_{\tau_{0}}^{1/f_{j}}(X)\right].
\end{eqnarray}
The proof is completed as follows:
\begin{eqnarray}
W_{x}\left[k(\tau_{0}) \mathcal{S}_{\tau_{0}}^{1/f_{j}}(X)\right]
&=& \lim_{n\rightarrow \infty} W_{x}\left[k(\tau_{0}) \mathcal{S}_{\tau_{0}}^{1/f_{j}}(X) : \tau_{0} \le \tau_{n}\right] \nonumber\\
&=& \lim_{n\rightarrow \infty} P_{x}^{1/f_{j}}\left[k(\tau_{0}) : \tau_{0} \le \tau_{n}\right] \nonumber\\
&=& P_{x}^{1/f_{j}}\left[k(\tau_{0}) : \tau_{0} < \infty \right] \nonumber\\
&=& P_{x}^{1/f_{j}}\left[k(\tau_{0})\right],
\end{eqnarray}
where the fact that $0$ is an exit boundary and $\infty$ is an inaccessible one for $X$ under $P_{x}^{1/f_{j}}$
is used.
\end{proof}  

\begin{proposition}
\label{prop:expression}
For $\tau>0$, we obtain (\ref{eq:Theta_j}) and (\ref{eq:Theta_0}).
\end{proposition}
\begin{proof}
We use the same notation as that in the proof of Lemma \ref{lem:key}.
For each $j>0$, Lemma \ref{lem:key} shows that
\begin{eqnarray}
\Psi_{j}(k,\lambda)
&=& \int_{0}^{\infty} \left(1-\frac{f(x+j)}{f(j)}\right) \lambda(x) 
P_{x}^{1/f_{j}}\left[k(\tau_{0})\right] \frac{h(f(x+j))}{f(x+j)} m_{1/f_{j}}(dx) \nonumber\\
&=& \int_{j}^{\infty} \left(1-\frac{f(x)}{f(j)}\right)\lambda(x-j) 
P_{x}^{1/f}\left[k(\tau_{j})\right] \frac{h(f(x))}{f(x)}m_{1/f}(dx),
\end{eqnarray}
in which we use the fact that the law of $X+j$ under $P_{x-j}^{1/f_{j}}$ is $P_{x}^{1/f}$.
The expression (\ref{eq:Theta_j}) follows by disintegration with respect to the variable of $k$.

Next, we derive (\ref{eq:Theta_0}) by letting $j \downarrow 0$.
First, we obtain
\begin{eqnarray}
\lim_{j \downarrow 0}\Psi_{j}(k,\lambda)  
&=&  \int_{0}^{\infty} \lambda(x) P_{x}^{1/f}\left[k(T_{0})\right] \frac{h(f(x))}{f(x)} m_{1/f}(dx) \nonumber\\
&=& \int_{0}^{\infty} \left(\int_{0}^{\infty} \lambda(x)\psi_{0}(\tau,x) \frac{h(f(x))}{f(x)} m_{1/f}(dx) \right) k(\tau) d\tau,
\end{eqnarray}
where $\psi_{0}(\cdot,x)$ is the density of $\tau_{0}$ under $P_{x}^{1/f}$.
Second, 
we obtain
\begin{eqnarray}
\lim_{j \downarrow 0}\Psi_{j}(k,\lambda)
= \lim_{j \downarrow 0} \int_{0}^{\infty} \Lambda_{j}(\tau,\lambda) k(\tau) d\tau 
= \int_{0}^{\infty} P_{0}^{f}[\lambda(X_{\tau})h(f(X_{\tau}))] k(\tau) d\tau
\end{eqnarray}
because 
$\Lambda_{j}(\tau,\lambda)$ is bounded
and Lemma \ref{lem:tightness} implies the following:
\begin{eqnarray}
\lim_{j \downarrow 0} \Lambda_{j}(\tau,\lambda) = P_{0}^{f}[\lambda(X_{\tau})h(f(X_{\tau}))].
\end{eqnarray}
Therefore, we obtain the following for $\lambda \in C_{K}((0,\infty))$ by disintegration:
\begin{eqnarray}
\int_{0}^{\infty} \lambda(x)\psi_{0}(\tau,x) \frac{h(f(x))}{f(x)} m_{1/f}(dx)
= P_{0}^{f}[\lambda(X_{\tau})h(f(X_{\tau}))]
\end{eqnarray}
for all $\tau$ owing to the continuity of $\psi_{0}$ and that of the right-hand side of the above equation with respect to $\tau$.
The conclusion follows by taking $\lambda \uparrow 1$.
\end{proof}  

Finally, we provide several expressions of $v^{*}(\tau,\infty)$.
Let $s_{\tau}(1/f(x),\cdot)$ be the density of $\inf_{u \le \tau}1/f(X_{u})$ with respect to $P_{x}^{1/f}$, restricted to the interval $(0,1/f(x))$,
and $v^{i,j}$ be the price of a knock-in forward contract with barrier level $f(j)$:
\begin{eqnarray}
v^{i,j}(\tau,y) := P_{f^{-1}(y)}^{f} [f(X_{\tau}):T_{j} \le \tau].
\end{eqnarray}
The forward price $v^{*}(\tau,\infty)$ can be expressed in terms of the density $s_{\tau}$ as well as the averages of the knock-in and knock-out forward prices.

\begin{proposition}
\label{prop:finite}
Suppose that $s_{\tau}(0,0) := \lim_{a,\eta\downarrow 0}s_{\tau}(a,\eta)$ exists in $[0,\infty]$ for $\tau > 0$, where $(a,\eta)$ approaches $(0,0)$ in any direction with $0 < a < \eta$, and the convergence is uniform.
Then, we obtain
\begin{eqnarray}
s_{\tau}(0,0)
= 2\lim_{j \downarrow 0}\int_{0}^{f(j)}v^{i,j}(\tau,y)\frac{dy}{f(j)}
= 2\lim_{j \downarrow 0}\int_{0}^{f(j)}v^{o,j}(\tau,y)\frac{dy}{f(j)}
= v^{*}(\tau,\infty).
\end{eqnarray}
\end{proposition}
\begin{proof}
The prices of knock-in and knock-out forward contracts are expressed as
\begin{eqnarray}
&&v^{i,j}(\tau,f(x))
= f(x) P_{x}^{1/f} [\tau_{j} \le \tau < \tau_{0}]
= f(x)\int_{0}^{1/f(j)} s_{\tau}(1/f(x),\eta)d\eta,\\
&&v^{o,j}(\tau,f(x))
= f(x) P_{x}^{1/f} [\tau < \tau_{j}]
= f(x)\int_{1/f(j)}^{\infty} s_{\tau}(1/f(x),\eta)d\eta.
\end{eqnarray}
The integrals of the proposition are
\begin{eqnarray}
2\int_{0}^{f(j)}v^{i,j}(\tau,y)\frac{dy}{f(j)}
&=& 2 \int_{0}^{f(j)} y \left(\int_{0}^{1/f(j)}s_{\tau}(1/y,\eta)d\eta\right) \frac{dy}{f(j)} \nonumber\\
&=& \int_{1/f(j)}^{\infty} \left(\int_{0}^{1/f(j)}s_{\tau}(y,\eta)\frac{d\eta}{1/f(j)} \right)\frac{2}{f(j)^{2}}\frac{dy}{y^{3}},
\end{eqnarray}
and
\begin{eqnarray}
2\int_{0}^{f(j)}v^{o,j}(\tau,y)\frac{dy}{f(j)}
= \int_{1/f(j)}^{\infty} \left(\int_{1/f(j)}^{y}s_{\tau}(y,\eta)\frac{d\eta}{y-1/f(j)}\right) \frac{y-1/f(j)}{f(j)} \frac{2dy}{y^{3}},
\end{eqnarray}
both of which tend towards $s_{\tau}(0,0)$ as $j$ approaches $0$.
Proposition \ref{prop:average_ko} with $h(y)=y$
and the convergence $\Theta_{j}(\tau) \rightarrow \Theta_{0}(\tau)$ lead to $v^{*}(\tau,\infty) = s_{\tau}(0,0)$.
\end{proof}

\section{Examples}
\label{sec:example}
In this section, we consider the cases of $s(x)=1/x$ and $f(x)=x^{-2\nu}$ for $\nu = \pm 1/2$.
The coordinate process $X$ is a three-dimensional Bessel process under $P_{x}^{s}$,
and the Schwarzian derivatives $S_{s}$ and $S_{f}$ are $0$.
Under the physical measure, the stock price is a three-dimensional Bessel process for $\nu=-1/2$,
and it is the reciprocal of a three-dimensional Bessel process for $\nu=1/2$.
We only consider European options with $h(x) = 1$ and $h(x)=x$; that is, a zero-coupon bond and a forward contract, respectively.
The price of a zero-coupon bond for the fundraiser is always $1$
because an equivalent local martingale measure exists for the fundraiser in the hypothetical market.

\subsection{Three-dimensional Bessel process: $\nu=-1/2$}
\label{sec:-1/2}
Arbitrage opportunities exist in this case.
The process $X$, which is also the stock price, is a Brownian motion that stops when it reaches $0$ under $P_{x}^{f}$
and a three-dimensional Bessel process under $\tilde{P}_{x}^{f}$,
and the GOPs are $G_{t} = X_{t}/x$ and $\tilde{G}_{t} = 1$.
Under $\tilde{P}_{x,j}^{f}$,
the stock price $X$ is a Brownian motion up to $\tau_{j}$
and $X_{\cdot + \tau_{j}}-j$ is a three-dimensional Bessel process that starts at $0$, independent of $\tilde{\mathcal{F}}_{\tau_{j}}^{0}$.
As mentioned after Lemma \ref{lem:comp},
the pathwise comparison results do not hold for $x$.
More precisely, let $X^{x,j}$ be defined by (\ref{eq:x_pf_xj}) with $X_{0}=x,J_{0} = j$.
Then, $X_{t}^{x,j} \le X_{t}^{x^{\prime},j}$ does not hold for all $t$, even if $x < x^{\prime}$,
which follows by considering the situation after $X^{x,j}$ reaches the level $j$
because $X^{x,j} > j$ always holds, whereas $X^{x^{\prime},j}$ eventually reaches $j$.

The price of a zero-coupon bond for the investor is calculated by
\begin{eqnarray}
W_{x}[X_{T} > 0] = N\left(\frac{x}{\sqrt{T}}\right)-N\left(-\frac{x}{\sqrt{T}}\right),
\end{eqnarray}
where $N$ is the distribution function of the standard normal distribution.
This is strictly less than $1$, which indicates an arbitrage.
The prices of the forward contract for each agent are 
\begin{eqnarray}
W_{x}[X_{T \wedge \tau_{0}} :X_{T} > 0] &=& x,\label{eq:f_bm}\\
\tilde{P}_{x,j}^{f}[X_{T}]
&=& x + 4\sqrt{T} \left(N^{\prime}\left(\frac{j-x}{\sqrt{T}}\right)
+\frac{j-x}{\sqrt{T}}N\left(\frac{j-x}{\sqrt{T}}\right) \right). \label{eq:f2_bm} 
\end{eqnarray}
Note that the delta of the forward contract for the fundraiser up to $\tau_{j}$;
that is, the derivative of (\ref{eq:f2_bm}) with respect to $x$, is 
$1-4N\left((j-x)/\sqrt{T}\right)$,
which takes values in $[-1,1)$ (in particular, $-1$ when $x=j$).
It is noteworthy that
the right-hand side of (\ref{eq:f2_bm}) converges to (\ref{eq:f_bm}) when $j \rightarrow -\infty$,
although we have defined $j$ in $[0,\infty)$.

\subsection{Reciprocal of three-dimensional Bessel process: $\nu=1/2$}
\label{sec:1/2}
This is a typical example of the bubble model.
The process $X$ is a three-dimensional Bessel process under $P_{x}^{f}$,
the stock price $Y$ is its reciprocal,
and the GOPs are $G_{t}=1$ and $\tilde{G}_{t} = x/J_{0}^{2} \cdot J_{t}^{2}/X_{t}$.
Although the equivalence between $\tilde{P}_{x}^{f}$ and $P_{x}^{(1/2)}$ on $\mathcal{\tilde{F}}_{T}^{0}$
follows from Lemma \ref{lem:equivalence},
it also follows from the fact that $1/\tilde{G}$ is a positive $P_{x}^{s}$ martingale because
\begin{eqnarray}
P_{x}^{(1/2)}\left[\sup_{t \le T} \left(\frac{1}{\tilde{G}_{T}}\right)^{2}\right]
\le
P_{x}^{(1/2)}\left[\sup_{t \le T} \frac{X_{t}^{2}}{x^{2}}\right]
\le
\frac{3}{x^{2}} W_{x}\left[\sup_{t \le T} X_{t}^{2} \right]
< \infty.
\end{eqnarray}
However, $\tilde{P}_{x}^{f}$ and $P_{x}^{(1/2)}$ are not equivalent on $\mathcal{\tilde{F}}_{\infty}^{0}$;
if they were, we would have $\tilde{P}_{x}^{f}[\Lambda_{t}^{*} < \infty] = 1$ for all $t \ge 0$.
Furthermore,
owing to
$\lim_{t \rightarrow \infty}X_{t}=\lim_{t \rightarrow \infty}J_{t}=\infty$,
we would have
\begin{eqnarray}
\lim_{t \rightarrow \infty} \tilde{G}_{t}
= \lim_{t \rightarrow \infty} \frac{x}{J_{0}^{2}}\frac{J_{t}^{2}}{X_{t}}
=
\lim_{t \rightarrow \infty} \frac{x}{J_{0}^{2}}J_{\Lambda_{t}^{*}}
= \infty,
\end{eqnarray}
$\tilde{P}_{x}^{f}$-a.s.,
which contradicts the fact that a continuous nonnegative supermartingale converges.

The price of a zero-coupon bond is $1$ for both agents because an equivalent local martingale measure exists.
The prices of a forward contract are calculated by
\begin{eqnarray}
P_{x}^{(1/2)}\left[\frac{1}{X_{T}}\right]
= \frac{1}{x} \left(N\left(\frac{x}{\sqrt{T}}\right)-N\left(-\frac{x}{\sqrt{T}}\right)\right) \label{eq:f1_bessel}
\end{eqnarray}
and
\begin{eqnarray}
\tilde{P}_{x,j}^{(1/2)}\left[\frac{1}{X_{T}}\right]
&=&
\frac{1}{x} \left(N \left(\frac{x-j}{\sqrt{T}}\right)-N \left(-\frac{x-j}{\sqrt{T}}\right)\right) \nonumber\\
&+&
\frac{2}{x} \int_{\frac{x-j}{\sqrt{T}}}^{\infty}
\left(\frac{j}{r\sqrt{T}-x+2j}\right)^{2}
\mathrm{e}^{-r^{2}/2} \frac{dr}{\sqrt{2\pi}}. \label{eq:f2_bessel}
\end{eqnarray}
As indicated in Theorem \ref{them:main} and Lemma \ref{lem:convergence},
the first term of the right-hand side of (\ref{eq:f2_bessel})
converges to (\ref{eq:f1_bessel}),
whereas the second term approaches $0$ as $j \rightarrow 0$.



\end{document}